\documentclass[pra,floatfix,amsmath,superscriptaddress,twocolumn]{revtex4}
\usepackage{amssymb}
\usepackage{graphicx}
\usepackage{graphics}
\usepackage{amsmath}
\usepackage{amsthm}
\usepackage{color}

\def\h{{\bf h}}
\def\g{{\bf g}}

\newcommand{\bea}{\begin{eqnarray}}
\newcommand{\eea}{\end{eqnarray}}

\def\bi{\begin{itemize}}
\def\ei{\end{itemize}}
\def\bc{\begin{center}}
\def\ec{\end{center}}

\def\C{\hbox{$\mit I$\kern-.7em$\mit C$}}
\def\R{\hbox{$\mit I$\kern-.6em$\mit R$}}
\def\N{\hbox{$\mit I$\kern-.6em$\mit N$}}

\def\ket#1{|#1\rangle}
\newcommand{\one}{\mbox{$1 \hspace{-1.0mm}  {\bf l}$}}
\def\tr{\mathrm{tr}}
\def\ket#1{\left| #1\right>}
\def\bra#1{\left< #1\right|}
\def\bk#1{\langle #1 \rangle}

\newtheorem{theorem}{Theorem}
\newtheorem{corollary}[theorem]{Corollary}
\newtheorem{lemma}[theorem]{Lemma}

\newtheorem{observation}[theorem]{Observation}

\begin{document}

\author{}

\author{J.I. de Vicente}
\affiliation{Departamento de Matem\'aticas, Universidad Carlos III de
Madrid, Legan\'es (Madrid), Spain}
\author{C. Spee}
\affiliation{Institute for Theoretical Physics, University of
Innsbruck, Innsbruck, Austria}
\author{D. Sauerwein}
\affiliation{Institute for Theoretical Physics, University of
Innsbruck, Innsbruck, Austria}
\author{B. Kraus}
\affiliation{Institute for Theoretical Physics, University of
Innsbruck, Innsbruck, Austria}
\title{Entanglement manipulation of multipartite pure states with finite rounds
of classical communication}

\begin{abstract}
We studied pure state transformations using local operations assisted by finitely many rounds of classical communication ($LOCC_{\N}$) in C. Spee, J.I. de Vicente, D. Sauerwein, B. Kraus, arXiv:1606.04418 (2016).
Here, we first of all present the details of some of the proofs and generalize the construction of examples of state transformations via $LOCC_{\N}$ which require a probabilistic step. However, we also present explicit
examples of SLOCC classes where any separable transformation can be realized by a protocol in which each step is deterministic (all-det-$LOCC_{\N}$). Such transformations can be considered as natural generalizations of
bipartite transformations. Furthermore, we provide examples of pure state transformations which are possible via separable transformations, but not via $LOCC_{\N}$.
We also analyze an interesting genuinely multipartite effect which we call locking or unlocking the power of other parties.
This means that one party can prevent or enable the implementation of LOCC transformations by other parties. Moreover, we investigate the maximally entangled set restricted
to $LOCC_{\N}$ and show how easily computable bounds on some entanglement measures can be derived by restricting to $LOCC_{\N}$.

\end{abstract}
\maketitle

\section{Introduction}
\label{secintro}

Quantum information theory has been successfully developed in the last decades and has shown that the non-classical features of quantum mechanics can be used to realize revolutionary technologies \cite{nielsenchuang}. These very promising applications include quantum communication, computation and simulation. Entanglement plays a crucial role in these quantum advantages, which has led to the development of entanglement theory \cite{reviews}. This theory studies the properties of entangled states and aims at understanding the ultimate possibilities and limitations of these states as a resource. However, while bipartite entanglement is fairly well understood, much more questions remain open in the multipartite realm. Although a handful of tasks have been already identified in this context, such as measurement-based quantum computation \cite{RaBr01}, metrology \cite{reviewmet} or secret-sharing \cite{SecretSh}, a deeper understanding of the complex structure of multipartite entangled states and of the fundamental protocols for their manipulation seems indispensable in order to find new truly multipartite applications of quantum information theory. Moreover, the entanglement properties of multipartite states are actively being exploited for the study of condensed-matter systems like, for example, in phase transitions \cite{AmFa08} or to devise numerical methods \cite{orus}.

The concept of transformations implementable by local operations assisted by classical communication (LOCC) plays a central role in entanglement theory \cite{reviews}. First of all, this is because, entanglement being a property of systems with many constituents, this is the most general form of manipulation for distant parties (that only share classical channels). Each party can implement locally any form of quantum dynamics, i.e.\ a completely positive map. However, the particular map to be implemented may depend on the outcomes of previous measurements carried out by other parties, an information that can be shared through the use of classical communication. Thus, LOCC protocols are made out of sequential rounds in which a given party implements locally a completely positive map via a quantum measurement using a particular Kraus decomposition and informs the others of the outcome of such measurement, which conditions the subsequent measurements to be carried out. With this, LOCC transformations provide the basis for all possible protocols for the manipulation of entangled states and, therefore, characterize the ultimate potential of quantum states for implementing quantum information tasks. Second and most importantly, entanglement theory is a resource theory where the free operations are those which can be realized precisely via LOCC. This means that LOCC convertibility induces the only operationally meaningful ordering in the set of entangled states. If $\rho$ can be transformed into $\sigma$ by LOCC, then $\rho$ is at least as useful as $\sigma$, as any application of the latter can be achieved by the former but not necessarily the other way around. This is because if the parties are provided with $\rho$, they can transform it to $\sigma$ at no cost and then carry on with the protocol. Thus, the study of LOCC transformations reveals the relative usefulness of the different quantum states and underpins the quantification of entanglement. Entanglement measures must be quantities that respect this ordering, i.e.\ that do not increase under these transformations \cite{reviews}.

Unfortunately, it turns out that the structure of LOCC maps is mathematically very subtle \cite{chitambar1}. Remarkably, it has been shown that there exist transformations involving ensembles of states that require infinitely
many rounds of classical communication \cite{chit11}. This result shows that the very intricate structure of successive rounds cannot be simplified a priori in the study of general protocols.
These difficulties have led to consider the larger class of separable (SEP) maps \cite{rains}. This set is strictly larger than LOCC \cite{sepnotlocc} and does not have any known operational interpretation but is
mathematically more tractable. Thus, convertibility under SEP operations is a necessary condition for LOCC convertibility. Notwithstanding, there are certain situations in which it is enough to consider simple LOCC
protocols to characterize LOCC convertibility. Reference \cite{LoPopescu} has shown that for transformations among pure bipartite states it suffices to consider protocols with one round of classical communication.
That is, it is enough that one party implements a quantum measurement upon whose result the other party only needs to implement a local unitary (LU) transformation. Interestingly, this simplification allowed to characterize
all LOCC transformations among pure bipartite states \cite{nielsen}. Moreover, it was later shown that in this setting SEP transformations are exactly as powerful as LOCC \cite{gheorghiu}. Recent works have studied LOCC
convertibility among pure multipartite states in the simplest possible settings of 3-qubit, 4-qubit and 3-qutrit states \cite{turgutghz,turgutw,dVSp13,SaSc15,SpdV16,HeSp15}. The techniques some of us have used there is to use the results of \cite{Gour} to characterize convertibility under SEP. Then, one can usually see that the possible SEP transformations can be indeed implemented by simple LOCC protocols that extend naturally those
which are sufficient in the bipartite case. All possible LOCC transformations among pure states identified so far fall into this category that we call all-deterministic $LOCC_{\N}$ (all-det-$LOCC_{\N}$).
These are protocols in which every round is deterministic. That is, upon every non-trivial measurement carried out by a party, the remaining parties can implement LUs conditioned on the outcome in such a way that at the end of the round one has the same state independently of the outcome of the measurement. This implies that every intermediate step in the protocol is a deterministic LOCC transformation as well. This extends the bipartite protocols with the only difference that in this case more parties might implement non-trivial measurements. Given these results, it would be tempting to conjecture that it suffices to consider all-det-$LOCC_{\N}$ protocols to characterize LOCC transformations among pure multipartite states, which could allow to solve this problem completely as it happened in the bipartite case. Still, it is worth mentioning that, contrary to the bipartite case, SEP transformations among pure multipartite states (3-qutrit states) have been identified that cannot be implemented by LOCC \cite{HeSp15}.
Some conditions that can allow to decide when a SEP protocol is implementable by LOCC have been derived in \cite{cohen} and references therein.

In \cite{short}, to which this article is the companion, we have studied LOCC transformations among pure multipartite states with finitely many rounds of classical communication.
We call these protocols $LOCC_{\N}$. First, this is a more realistic scenario of protocols implementable in practice.
Second, this is the natural approach to complement our previous efforts that relied on SEP. While SEP approximates the set of LOCC maps from the outside,
$LOCC_{\N}$ considers a very general set of protocols that approximate LOCC maps from the inside. Actually, it is not even known if infinite-round LOCC protocols are strictly
more powerful than $LOCC_{\N}$ for pure-state manipulation. Moreover, we provide for the first time general techniques to address convertibility under this very general class of LOCC maps.
The key observation in this case is that all $LOCC_{\N}$ protocols must terminate and, hence, every branch of the LOCC protocol has to finish with a deterministic measurement by one party.
This puts severe constraints on the possible transformations as we have shown in \cite{short}. Interestingly, our techniques apply to a very general class of multipartite states of arbitrary
dimensions and arbitrary number of subsystems. In particular, it has been shown that almost every $n$-qubit state ($n>3$) \cite{bookWallach} and almost all three-qutrit states \cite{HeSp15} belong to this class.

In Sec. \ref{sec:results} we give a detailed
outline of this paper and concisely summarize our results. On the one hand, we provide full
proofs of many results outlined in \cite{short} and on the other hand, we generalize and extend the analyis of $LOCC_{\N}$ transformations of pure states. In particular, we present general results on these transformations and
analyze a genuinely multipartite phenomenon which we called locking and unlocking the power of others in \cite{short}. Moreover, we investigate the maximally entangled set (MES) \cite{dVSp13} if restricted to $LOCC_{\N}$
and derive easily computable bounds on certain entanglement measures. Furthermore, we identify classes of states where any pure state transformation within SEP can be realized by LOCC
via an all-det-$LOCC_{\N}$ protocol. In contrast to that, we provide a general construction of state transformations that show that all-det-$LOCC_{\N}$ protocols are not sufficient to implement
any $LOCC_{\N}$ pure state transformation.

\section{Outline and results}
\label{sec:results}
We investigate $LOCC_{\N}$ transformations among truly $n$-partite entangled pure states $\ket{\Psi}, \ket{\Phi} \in \C^{d_1}\otimes \ldots \otimes \C^{d_n}$
whose single-subsystem reduced states have full rank. We consider states that are elements of the same stochastic LOCC (SLOCC) class, represented by a state $\ket{\Psi_s}$. Moreover, we consider SLOCC classes for which the
local stabilizer of $\ket{\Psi_s}$, i.e. the set of all local invertible matrices $S = S^{(1)} \otimes \ldots \otimes S^{(n)}$ for which $S\ket{\Psi_s} = \ket{\Psi_s}$,
contains only finitely many elements (see Sec. \ref{secprel} for details). These classes include, e.g., almost all three-qutrit and almost all $n$-qubit states, for $n>3$. We call a
state $\ket{\Phi}$ reachable via $LOCC_{\N}$,
in short $LOCC_{\N}$-reachable, if there is a (LU-inequivalent) state $\ket{\Psi}$ from which $\ket{\Phi}$ can be obtained
via $LOCC_{\N}$. In \cite{short} we provided a very simple characterization of all $LOCC_{\N}$-reachable states (see also Sec. \ref{secreachable}). A particularly simple class of $LOCC_{\N}$ protocols are all-det-$LOCC_{\N}$.
There, in the first step (round) e.g. $A$ performs a measurement and the other parties apply, depending on the measurement outcome LUs to transform the initial state, $\ket{\Psi}$ into a
state $\ket{\Psi_1}$ deterministically. In the second step e.g. $B$ applies a measurement (and the other parties LUs) to transform $\ket{\Psi_1}$ into $\ket{\Psi_2}$, etc.
The protocol proceeds in this way until the final state is reached. That is, in each step a deterministic transformation is realized. As mentioned before, all previously identified pure state transformations are
(up to our knowledge) of this simple form. Moreover, in the bipartite setting, any transformation can be realized by a all-det-$LOCC_{\N}$. In fact, for pure state transformations we have in the bipartite case
that $all-det-LOCC_{\N} = LOCC_{\N} = LOCC = SEP$. However, in the multiparite setting this simple relations change to $all-det-LOCC_{\N} \subsetneq LOCC_{\N} \subseteq LOCC \subsetneq SEP$, highlighting the difference
between bipartite and multipartite entanglement. The inequivalence between $LOCC$ and $SEP$ has been shown in \cite{HeSp15}, whereas the fact that not any $LOCC_{\N}$ is of the simple form of a all-det-$LOCC_{\N}$ is shown
in \cite{short}. Here, apart from providing the detailed proofs of the results presented in \cite{short}, we also present a general method of constructing examples of pure state transformations for which
more sophisticated $LOCC_{\N}$ protocols are required than all-det-$LOCC_{\N}$. Moreover, we analyze important features which only occur in the multipartite setting, characterize important sets of states and derive bounds on certain entanglement measures using the characterization of $LOCC_{\N}$--reachable states.
Below we list and explain the results presented here in more detail.\\

\paragraph{States reachable via $LOCC_{\N}$ (Sec. \ref{secreachable})}\hfill\\
We first review the characterization of $LOCC_{\N}$-reachable states (Theorem 1) in SLOCC classes with finite stabilizer given in \cite{short}. A corollary of Theorem 1 stated in \cite{short} is that almost no $n$-qubit
state ($n>3$) is reachable via $LOCC_{\N}$.
 In Sec. \ref{secreachable} of this work we present a detailed
proof of this result. We then use Theorem 1 to find an example of a four-qubit state transformation that can be implemented via SEP, but not via $LOCC_{\N}$.
This shows that the result of \cite{SaSc15} that $LOCC_{\N} = SEP$ for generic four-qubit states does not generalize to all four-qubit states.\\

\paragraph{States convertible via all-det-$LOCC_{\N}$ (Sec. \ref{secconvertible})}\hfill\\
All-det-$LOCC_{\N}$ protocols are considerably
more structured than general LOCC protocols. Yet,
all LOCC transformations among pure multipartite (and bipartite) states, that we are aware of, can also be performed via an all-det-$LOCC_{\N}$ protocol. In \cite{short} we gave a characterization of
all-det-$LOCC_{\N}$ transformations within SLOCC classes with finite stabilizer. In
\ref{secreachable} of this work we provide a full analysis of
these results and provide rigorous proofs thereof.\\

\paragraph{Locking and unlocking of power (Sec. \ref{seclocking})}\hfill\\
We analyze an interesting genuinely multipartite effect that we have only briefly outlined in \cite{short}: locking and unlocking the power of other parties. This means that one party, $i$, can
prevent or enable the transformation of a pure state via a one-round all-det-$LOCC_{\N}$ protocol in which only one party, $j$, acts nontrivially. That is, party $i$ can prevent or enable party $j$ to
start an all-det-$LOCC_{\N}$ transformation. We investigate under which conditions these phenomena can occur. Moreover, we provide explicit examples of unlocking the power. We also show that there are classes
of states in which unlocking is not possible and provide
explicit examples thereof.\\

\paragraph{The Maximally entangled set under $LOCC_{\N}$ (Sec. \ref{secmes})}\hfill\\
In \cite{dVSp13} the concept of the maximally entangled set (MES) has been introduced, which can be considered as the natural generalization of the maximally entangled state. The MES is defined as the minimal set of states
of a quantum system from which all other truly $n$-partite entangled pure states can be obtained deterministically
via LOCC. While the bipartite MES contains (up to LUs) only the maximally entangled state, it can contain infinitely many states and even be of full measure \cite{dVSp13} in the multipartite case. In Sec. \ref{secmes}
we analyze the MES if the set of operations is restricted to $LOCC_{\N}$. Moreover, we use the characterization of all-det-$LOCC_{\N}$ to identify all states in the MES restricted to
$LOCC_{\N}$ that are convertible via all-det-$LOCC_{\N}$ protocols to an LU-inequivalent state.\\

\paragraph{Estimation of entanglement measures (Sec. \ref{secmes})}\hfill\\
In \cite{ScSa15} two operational entanglement measures have been introduced, the source and the accessible
entanglement. The former characterizes the difficulty of generating the state at hand, and the latter the potentiality of a state to generate other states via LOCC.
In this work, we show that one can define similar functions that quantify the
resourcefulness of a quantum state under all-det-$LOCC_{\N}$.
Clearly, the difficulty of generating a
state via all-det-$LOCC_{\N}$ is at least as high as via LOCC as all-det-$LOCC_{\N} \subseteq LOCC$.
Similarily, the potentiality of a state to generate other states via all-det-$LOCC_{\N}$ is at most as high as via $LOCC$. In Sec. \ref{secmes} we use these insights in combination with our characterization of all-det-$LOCC_{\N}$ to derive bounds on the accessible and the source entanglement.\\

\paragraph{Examples of SLOCC classes where all-det-$LOCC_{\N}$ protocols are sufficient (Sec. \ref{secpaulis})}\hfill\\
As explained above, it holds that $all-det-LOCC_{\N} \subsetneq LOCC_{\N} \subseteq LOCC \subsetneq SEP$. However, in Sec. \ref{secpaulis} we explicitly construct SLOCC classes of $2^m$-qubit states ($m\geq2$) for which $all-det-LOCC_{\N} = SEP$ holds.
Hence, all $LOCC$ transformations within these classes can be characterized using the results on all-det-$LOCC_{\N}$ of Sec. \ref{secconvertible}. For example, this makes it easy to derive entanglement measures for
pure states for these classes, as we explain in Sec. \ref{secpaulis}.\\

\paragraph{All-det-$LOCC_{\N}$ protocols are not sufficient for $LOCC_{\N}$ pure-state transformations (Sec. \ref{secexample})}\hfill\\
As mentioned before, all LOCC transformations of multipartite pure states (including the bipartite case) that have, to our knowledge, been identified so far, can be realised by a relatively simple all-det-$LOCC_{\N}$ protocol.
However, in \cite{short} we have shown that this is not always the case, i.e. that $all-det-LOCC_{\N} \subsetneq LOCC_{\N}$. More precisely, we have
considered a particular example of a pure state transformation that can be performed via $LOCC_{\N}$ but requires a non-deterministic intermediate step, and hence cannot be performed via all-det-$LOCC_{\N}$. In Sec. \ref{secexample}
of this work
we discuss a general construction that allows to obtain such examples.
These results show that the transformation of multipartite entanglement can require LOCC protocols that are more sophisticated than the bipartite protocols. Moreover, it shows again that a characterization of all
multipartite LOCC transformations is very complex and may stay elusive.

\section{Preliminaries}
\label{secprel}

We consider transformations among pure truly $n$-partite entangled states of the same dimensions. That is, both the input and target state fulfill
$\ket{\Psi}, \ket{\Phi} \in \C^{d_1}\otimes \ldots \otimes \C^{d_n}\equiv \mathcal{H}=\otimes_i\mathcal{H}_i$ (with $d_i$ denoting the dimension of subsystem $i$) and are entangled in the same dimensions,
i.e. the single-subsystem reduced density matrices have full rank. Moreover, we consider states in the same stochastic LOCC (SLOCC) class \cite{slocc}, which means that there exist invertible matrices $A_i$ for $i\in \{1,\ldots ,n\}$ such that $\ket{\Psi}\propto A_1\otimes \ldots \otimes A_n\ket{\Phi}$.

Any trace-preserving completely positive map $\Lambda$ acting on states $\rho$ on $\mathcal{H}$ admits a Kraus representation
\begin{equation}\label{cpmap}
\Lambda(\rho)=\sum_iX_i\rho X_i^\dag,
\end{equation}
where the Kraus operators $\{X_i\}$ fulfill the normalization condition $\sum_iX_i^\dag X_i=\one$. Rather than maps we usually refer to quantum measurements represented by the positive-operator valued measure (POVM) elements $\{X_i^\dag X_i\}$. The map $\Lambda$ is said to be SEP if all the (properly normalized) Kraus operators are product operators on $\mathcal{H}$, i.e.\
\begin{equation}\label{sepmap}
X_i=X_i^{(1)}\otimes \ldots \otimes X_i^{(n)}\quad\forall i,
\end{equation}
where $X_i^{(j)}$ acts on $\C^{d_j}$. The map $\Lambda$ is said to be $LOCC_{\N}$ with $m\in\N$ rounds of classical communication if we can split the index of the map into a multi-index
with $m$ entries $i=(i_1\cdots i_m)$ and the Kraus operators factorize according to
\begin{widetext}
\begin{equation}\label{loccmap}
X_i=X_{(i_1\cdots i_m)}=\prod_{k=1}^m U_{i_k}^{(1)}(\{i_j\}_{j<k})\otimes\cdots \otimes X_{i_k}^{(s)}(\{i_j\}_{j<k})\otimes\cdots\otimes U_{i_k}^{(n)}(\{i_j\}_{j<k}),
\end{equation}
\end{widetext}
where the order of the product in $k$ is from right to left. This means that the action of the parties at each round $k$ of the protocol (i.e.\ the corresponding map to be implemented) depends on the previous indices $\{i_j\}_{j<k}$.
Moreover, the party or parties, $s$, that act non-unitarily in round $k$ depend as well on $\{i_j\}_{j\leq k}$, i.e.\ $s=s(\{i_j\}_{j\leq k})$ (which we do not explicitly write to ease the notation) \footnote{Note that the fact that $s$ may depend on $i_k$ reflects that in the round $k$ several parties might act with certain probability.}. These parties have to implement a proper POVM and, hence, the normalization
\begin{equation}
\sum_{i_k}(X_{i_k}^{(s)}(\{i_j\}_{j< k}))^\dag X_{i_k}^{(s)}(\{i_j\}_{j< k})=\one
\end{equation}
has to hold at every branch of the protocol corresponding to the round $k$ (i.e.\ for all values of $\{i_j\}$ such that $j<k$). The round finishes by the remaining parties
implementing a unitary transformation $U_{i_k}^{(l)}(\{i_j\}_{j<k})$ ($l\neq s$) determined not only by all previous indices $\{i_j\}_{j<k}$ but also by the outcome at the present round $i_k$.
We say that a transformation from $\ket{\Psi}$ to $\ket{\Phi}$ can be implemented by SEP or $LOCC_{\N}$ whenever there exist maps as specified above such that $\Lambda(\ket{\Psi}\bra{\Psi})=\ket{\Phi}\bra{\Phi}$. Notice that LU transformations,
\begin{equation}
\ket{\Phi}=U^{(1)}\otimes\cdots\otimes U^{(n)}\ket{\Psi},
\end{equation}
are invertible LOCC transformations that are always possible and do not change the entanglement of a state. Thus, we only consider transformations among different LU-equivalence classes even if we refer to states. In this way, whenever there exist (SEP or $LOCC_{\N}$) protocols such that $\Lambda(\ket{\Psi}\bra{\Psi})=\ket{\Phi}\bra{\Phi}$ and the states are not LU-equivalent we say that the state $\ket{\Psi}$ is convertible and the state $\ket{\Phi}$ is reachable (by SEP or $LOCC_{\N}$).

The general class of multipartite states that we are going to consider here is given by the following condition. They belong to an SLOCC class such that it has a representative $\ket{\Psi_s}$ that has a \emph{finite} local stabilizer group \cite{Gour}. The local stabilizer of a state $\ket{\Psi_s}$, $S_{\Psi_s}$, is the group of all product invertible matrices $S$ such that
\begin{equation}
S\ket{\Psi_s}=\ket{\Psi_s}\textrm{ with }S=S^{(1)} \otimes S^{(2)}\otimes \ldots \otimes S^{(n)}.
\end{equation}
It has been moreover shown in \cite{Gour} that if the stabilizer is finite, one can always find a representative of the class $\ket{\Psi_s}$ such that all elements of the stabilizer (which we also call symmetries) correspond to LUs. Thus, without loss of generality, we always consider this to be the case. Every state which is in the same SLOCC class as $|\Psi_s\rangle$ can then be written (ignoring normalization) as $g\ket{\Psi_s}$, with $g=\otimes_{i=1}^n g_i$ where $g_i\in GL(d_i)$ $\forall i$. Obviously, the operator $g$ is not unique for a given LU-equivalence class. For this reason, we consider the positive operators $G_i=g_i^\dag g_i$ that lead to $G=\otimes_{i=1}^n G_i$. With this, the operator $G$ is uniquely given by the LU-equivalence class up to conjugation by elements of the stabilizer group $S_k^\dag G S_k$. Without loss of generality, the normalization of the states is chosen such that the operators $\{G_i\}$ have all unit trace. In the particular case of $n$-qubit states, to refer to the states we often use the $\R^3$ vectors $\{\g_i\}$ arising from the Bloch representation
\begin{equation}
G_i=\frac{\one}{2}+\g_i\cdot\vec\sigma,
\end{equation}
where $\vec\sigma=(\sigma_1,\sigma_2,\sigma_3)$ and $\{\sigma_i\}$ represent the Pauli matrices (we also use $\sigma_0=\one$). Notice that it must hold that $||\g_i||<1/2$ $\forall i$ in order to guarantee the positiveness of $G_i$.

Let us also review here the results of \cite{Gour} that characterize SEP conversions among pure states. Whenever we study a transformation inside a given SLOCC class, we always denote by $g\ket{\Psi_s}$ the initial state and by $h\ket{\Psi_s}$ the final state (the operators $\{H_i\}$ and vectors $\{\h_i\}$ are given on the analogy to the above definitions). A transformation by SEP among pure states in an SLOCC class with unitary stabilizer group $S_{\Psi_s}$ is possible if and only if (iff) there exists a probability distribution $\{p_k\}$ over the finite set of symmetries such that
\begin{equation}\label{condsep}
\sum_kp_kS_k^\dag HS_k=G.
\end{equation}
This shows the fundamental role played by the stabilizer group in this context. The intuition behind this result is that the only way to obtain a SEP map such that
\begin{equation}
\Lambda\left(\frac{g\ket{\Psi_s}\bra{\Psi_s}g^\dag}{\tr
(g\ket{\Psi_s}\bra{\Psi_s}g^\dag)}\right)=\frac{h\ket{\Psi_s}\bra{\Psi_s}h^\dag}{\tr(h\ket{\Psi_s}\bra{\Psi_s}h^\dag)}
\end{equation}
is that the Kraus operators (as given by Eqs.\ (\ref{cpmap})-(\ref{sepmap})) fulfill
\begin{equation}
X_k=\sqrt{p_k}h_1 S_k^{(1)} g_1^{-1}\otimes \ldots \otimes h_n S_k^{(n)} g_n^{-1},
\end{equation}
for some symmetries $\{S_k\}\in S_{\Psi_s}$. Equation (\ref{condsep}) ensures then the proper normalization of the map. Notice that the set of operators $H$ which can be reached by SEP from $G$ is convex. That is, the states which are reachable from a state $g\ket{\Psi_s}$ are defined by the convex set
$$\{H: \exists \{p_k\geq0\}, \sum_k p_k=1: \sum_k p_k S_k^\dagger H S_k =G\}.$$

To conclude this preliminary section let us discuss the generality of the states considered here, i.e.\ those that belong to an SLOCC class with finite (and hence unitary) stabilizer.
First of all, it has been shown that generic $n$-qubit states with $n>3$ belong to SLOCC classes with this property \cite{bookWallach}. In the more general case of multiqudit states some of the SLOCC
classes also possess a representative which has only finitely many local symmetries. Moreover, as for instance in the case of 3-qutrit states, the union of these SLOCC classes can be generic too \cite{BrLu04}.
Notwithstanding, considering arbitrary SLOCC classes, it should be stressed that it is not clear, even if the stabilizer is known to be finite, whether it is the case that there exist non-trivial symmetries
beyond $S=\one$. For SLOCC classes with a trivial stabilizer the aforementioned result on SEP convertibility allows to prove that no SEP (and, therefore, no LOCC) conversion is possible among pure states. Thus, besides the known examples of 4-qubit and 3-qutrit states in Sec. \ref{secpaulis} we provide more examples of SLOCC classes for $2^m$--qubit states ($m\geq 2$) that have a non-trivial finite stabilizer.

\section{General results on convertibility under $LOCC_{\N}$}\label{secgeneral}

To start, we review the results on $LOCC_{\N}$-convertibility presented in the companion paper \cite{short} and we present fully-detailed proofs of claims outlined therein. We also discuss more extensively the implications of some of these investigations and provide explicit examples for the phenomenon that the action of one party can allow or prevent another party to implement a deterministic transformation.
\subsection{Reachable states}\label{secreachable}

One of our main results in \cite{short} is to characterize all reachable states via $LOCC_{\N}$ in the considered SLOCC classes as it is stated in the following theorem.

\begin{theorem} \label{theorem_1}
A state $\ket{\Phi}\propto h \ket{\Psi_s}$ is reachable via $LOCC_{\N}$, iff there exists $S \in S_{\Psi_s}$ and a $j\in \{1,\ldots,n\}$ such that:
\bi \item[(i)] For any $i\neq j$ $[H_i, S^{(i)}]=0$ and
\item[(ii)] $[H_j, S^{(j)}]\neq 0$.\ei
\end{theorem}

The proof of this theorem is given in \cite{short}. The idea behind it is somehow similar to the intuition given above on the condition under SEP. Since the protocol contains a finite number of rounds, in every branch there must exist a last measurement by one party that maps deterministically the state in this branch to the target state so that the process terminates. If we denote this second-to-last state in one branch by $g\ket{\Psi_s}$, the measuring party, say $j$, can only implement a POVM with local measurement operators $X^{(j)}_k\propto h_jS_k^{(j)}g_j^{-1}$, which requires condition (ii). The fact that the other parties can map by LUs the different outcomes to the final state requires condition (i).

Theorem \ref{theorem_1} easily reveals whether a state is reachable or not once the unitary stabilizer $S_{\Psi_s}$ is determined. In \cite{SpdV16,HeSp15} we have provided general means to achieve this task. Moreover, as we discuss below our techniques also allow to answer in many instances which states can be transformed into the reachable states. However, interestingly, it is now very intuitive to see that very few states are going to have the property of being reachable. Notice, that this is only the case for general states $\bigotimes_ih_i|\Psi_s\rangle$ if all but one of the $\{H_i\}$ commute with the local part of one symmetry $S\in S_{\Psi_s}$, which is a very restrictive condition. Actually, in \cite{short} we give the following corollary, whose proof we detail here.

\begin{corollary}
The set of $n$-qubit states ($n>3$) which are reachable via a $LOCC_{\N}$ protocol is of measure zero. \end{corollary}
\begin{proof}
As stated above, for $n$-qubit states of more than three parties our SLOCC classes are generic and, hence, Theorem 1 applies. In the Bloch representation picture the action of the SU(2) unitaries $S_k^{(i)}$ translate to SO(3) rotations. Thus, $[H_i, S^{(i)}]=0$ can only hold if $\h_i$ lies on the axis of the corresponding rotation (or is zero). Thus, since there are only finitely many symmetries, the sets of vectors $\{\h_i\}$ fulfilling condition (i) in Theorem 1 are of measure zero in the Bloch ball. It just remains to take into account that two sets of vectors correspond to the same LU-equivalence class iff $H'=S^\dag HS$ for some $S\in S_{\Psi_s}$. With this, one can easily see that condition (i) in Theorem 1 is fulfilled either in every LU-equivalent representation or in none. Thus, the generic sets of vectors that do not fulfill condition (i) cannot do so in another LU-equivalent representation.
\end{proof}

Note again that all results in this paper obviously also apply to SLOCC classes whose stabilizer is trivial. Among states in these classes LOCC transformations are not possible at all.
However, Corollary 2 shows that almost no $n$-qubit state is reachable via $LOCC_{\N}$ independently of whether their stabilizer is trivial or not
\footnote{After completing this manuscript it was proven that generic $n$-qubit-states ($n\geq 5$) indeed have a trivial stabilizer \cite{GoKr16}.
This allows for an alternative proof of our Corollary 2 that also extends to infinite-round protocols.}. Moreover, this result applies also to all SLOCC families
(generic or not) of multipartite states of arbitrary dimension with a non-trivial finite stabilizer, which are clearly the only a priori relevant classes with a finite stabilizer in this context. Using the same arguments as in the proof of Corollary 2, it is straightforward to see that the set of states which are reachable in a given non-trivial SLOCC class is of measure zero there. In particular, whenever our considered SLOCC classes are generic (e.g.\ 3-qutrit states) we have that almost no state of the given number of parties and local dimensions is reachable in the full set of states.

It might be illuminating to compare the result of Corollary 2 to the generic 3-qubit case (GHZ SLOCC class), where almost every state is reachable by an $LOCC_{\N}$ protocol \cite{turgutghz,dVSp13}. Actually, in this case it can be seen that it is sufficient to consider the protocols used in the proof of Theorem 1 by just considering a finite subset of unitary symmetries $\tilde S_{GHZ}$ of the full non-finite (and not even compact) \footnote{Note that in \cite{Gour} it has been shown that if the stabilizer is not finite it is not even compact.}  stabilizer $S_{GHZ}$. However, states are generically reachable in this case because one can always find a non-unitary symmetry $S\in S_{GHZ}$ such that almost every LU-equivalence class can be put in a form such that $S^\dag HS$ fulfills the commutation relations of the conditions of Theorem 1 with some unitary symmetry in $\tilde S_{GHZ}$.

It is also interesting to compare the power of $LOCC_{\N}$ as given by Theorem 1 with that of SEP. Without using the fact that $LOCC_{\N}\subset$ SEP, it should be clear that the conditions of Theorem 1 are sufficient for SEP convertibility since in this case the states $g \ket{\Psi_s}$ and $h \ket{\Psi_s}$ with
\begin{equation}
G=H_1\otimes\cdots\otimes\left(pH_j+(1-p)(S^{(j)})^\dag H_jS^{(j)}\right)\otimes\cdots\otimes H_n
\end{equation}
for some $p\in(0,1)$ clearly satisfy the condition of Eq.\ (\ref{condsep}). However, one might question whether these conditions are also sufficient for SEP. It turns out that for a large number of SLOCC classes this is indeed the case \cite{dVSp13,SpdV16}, and we then have in these cases that the set of reachable states is the same under $LOCC_{\N}$, LOCC and SEP. However, this is not always the case for every SLOCC class and this is what has allowed to prove in \cite{HeSp15} that there exist SEP transformations among pure states which are not achievable by LOCC (even if infinitely many rounds of communication are considered) in the case of 3-qutrit states. Here, we use Theorem 1 to provide a different example of a state conversion that is possible by SEP but not by $LOCC_{\N}$ (it remains unclear whether it is possible by LOCC or not). For this example we use the 4-qubit L-state SLOCC class. This family will be considered in different sections throughout this paper and, therefore, we summarize its properties in the following subsection besides providing the aforementioned result.

\subsubsection{The L-state SLOCC class and an example of SEP transformations which are not implementable with $LOCC_{\N}$}\label{seclstate}

The L state is a 4-qubit state given by
\begin{equation}\label{lstate}
|L\rangle=\frac{1}{\sqrt{3}}(|\phi^-\rangle|\phi^-\rangle+e^{i\pi/3}|\phi^+\rangle|\phi^+\rangle
+e^{i2\pi/3}|\psi^+\rangle|\psi^+\rangle),
\end{equation}
where we use the standard 2-qubit Bell basis
\begin{equation}\label{bell}
\ket{\phi^\pm}=\frac{1}{\sqrt{2}}(\ket{00}\pm \ket{11}),\; \ket{\psi^\pm}=\frac{1}{\sqrt{2}}(\ket{01}\pm \ket{10}).
\end{equation}
This state was considered in \cite{GoWa}, where it was noticed that it has interesting properties: it maximizes among 4-qubit states the average entanglement across bipartitions for several general bipartite entanglement measures. Remarkably, the L state and its corresponding SLOCC class also show a very interesting behaviour under multipartite entanglement manipulation. Its local stabilizer is simple but at the same time rich enough so that this class seems to be the perfect candidate to seek for out-of-the-ordinary features that break the general rule \cite{SpdV16}. As mentioned above, we use this class here to provide examples of SEP transformations that cannot be implemented by $LOCC_{\N}$ but we consider it again in Secs.\ \ref{seclocking} and \ref{secexample} to show other peculiarities of multipartite LOCC manipulation.

The stabilizer of the L state is given by 12 elements \cite{SpdV16},
\begin{equation}\label{Lstabilizer}
S_L=\{\{\one, U, U^2 \}\times\{\sigma_i\}_{i=0}^3\}^{\otimes4},
\end{equation}
where $U=\exp(i\frac{\pi}{4}\sigma_2)\exp(i\frac{\pi}{4}\sigma_1)$. It might be more illuminating to consider the SU(2) conjugation $UG_iU^\dag$ as an SO(3) rotation $O_U\g_i$ in the Bloch picture, where it turns out that
\begin{equation}
O_U=\left(
      \begin{array}{ccc}
        0 & 1 & 0 \\
        0 & 0 & 1 \\
        1 & 0 & 0 \\
      \end{array}
    \right).
\end{equation}
That is, conjugation by $U$ amounts to a cyclic permutation of the entries of the Bloch vector, or, in other words, it corresponds to a rotation of $2\pi/3$ around the axis given by $(1,1,1)$. Similar considerations apply to $U^2=-U^\dag$.

Here and in Sec. \ref{secmes} we use the twirling operation. Given a finite unitary group $S=\{S_k\}$ it is defined by
\begin{equation}\label{twirling}
T_{S}(H)=\frac{1}{|S|}\sum_kS_kHS_k^\dag.
\end{equation}
It is straightforward to verify that $T_S(H)=H$ iff $[H,S_k]=0$ $\forall k$ and its kernel is given by the orthogonal complement of the commutant of $S$. Inspired by the construction of \cite{HeSp15}, it is a simple exercise to see that
\begin{equation}
T_{S_L}\left(H_1\otimes H_2\otimes \one \otimes \one \right)\propto \one^{\otimes4}
\end{equation}
whenever $\h_1\cdot\h_2=0$. By virtue of Eq.\ (\ref{condsep}), this shows that $|L\rangle$ can be transformed by SEP into any state $h_1\otimes h_2\otimes\one\otimes\one|L\rangle$ whose Bloch vectors fulfill the above condition. However, if we choose $\h_1$ and $\h_2$ so that they do not point in the X, Y or Z direction nor have all their entries equal up to sign, which is certainly possible, we have then that
\begin{equation}
[H_j,\sigma_i],[H_j,U],[H_j,U\sigma_i]\neq0,\quad j=1,2,i=1,2,3.
\end{equation}
That is, $[H_j,S_k^{(j)}]\neq0$ for $j=1,2$ and $\forall S\in S_L$ (except $S=\one^{\otimes4}$). Therefore, the conditions of Theorem 1 are not met and all such states are not reachable by $LOCC_{\N}$ from any other state in the L-state class despite being so by SEP from $|L\rangle$ \footnote{Note that this is one of the reasons why we mentioned in \cite{SpdV16} that the L--class resembles the generic 3-qutrit case, where we proved that in fact the corresponding 3-qutrit states are not even reachable via LOCC (including infinitely many rounds of communication) despite being reachable via SEP.}.

\subsection{Convertible states}\label{secconvertible}

In the companion paper \cite{short} we have also discussed which states could then be converted to the characterized set of reachable states. Interestingly, in the proof of sufficiency for Theorem 1 we have shown in \cite{short} that every reachable state can be obtained from some other state by a very simple protocol: $LOCC_j$. This is the set of LOCC protocols which consist of only one round of classical communication and in which only party $j$ acts non-unitarily (i.e.\ in Eq.\ (\ref{loccmap}) $m=1$ and $s=j$). $LOCC_j$ maps are then the building blocks of all-det-$LOCC_{\N}$ protocols, that we have referred to in the introduction. The latter can be defined as those $LOCC_{\N}$ protocols that can be obtained as a concatenation of a finite number of $LOCC_j$ maps in which in every step the non-unitary party $j$ can be different. We stress here again that all LOCC transformations among pure states determined so far belong to this class of protocols. It is thus a very natural question to ask which states are convertible via $LOCC_j$ for some party $j$ as, together with Theorem 1, this then characterizes all states which can occur in all-det-$LOCC_{\N}$ transformations. This is the subject of the following lemma that we have outlined in \cite{short} and whose full proof we present for the sake of readability in Appendix \ref{lemma3}. Without loss of generality we consider $LOCC_1$.

\begin{lemma} \label{lemmaconv}
A state $\ket{\Psi}\propto g\ket{\Psi_s}$ is convertible via $LOCC_1$ iff there exist $m$ symmetries $S_k\in S_{\Psi_s}$, with $m> 1$ and $H\in {\cal B} ({\cal H}_1)$, $H >0$ and $p_k> 0$ with $\sum_{k=1}^m p_k=1$ such that the following conditions hold
\bi \item[(i)] $[G_i,S_k^{(i)}]=0$ $\forall i>1$ and $\forall k \in \{1,\ldots,m\}$
\item[(ii)] $G_1=\sum_{k=1}^m p_k (S_k^{(1)})^\dagger H S_k^{(1)}$ and $H \neq S^{(1)} G_1 {S^{(1)}}^\dagger$ for any $S \in S_{\Psi_s}$ fulfilling (i).
\ei
\end{lemma}
Note that the second condition in (ii) ensures that trivial transformations are excluded.
A state $g\ket{\Psi_s}$ fulfilling the premises of Lemma 3 can be transformed by $LOCC_1$ into $h_1\otimes g_2 \otimes \ldots \otimes g_n\ket{\Psi}$ by party 1 performing the POVM with measurement operators $\{A_k= \sqrt{p_k} h_1 S_k^{(1)} g_1^{-1}\}_{k=1}^m$ (which satisfies $\sum_kA_k^\dag A_k=\one$ due to condition (ii)). After this measurement, the parties hold the states $h_1S_k^{(1)}\otimes g_2\otimes\cdots\otimes g_n|\Psi_s\rangle$ depending on the outcome $k$. Upon reception of the measurement outcome through the classical channel, each party $i\neq 1$ applies the unitary $U_k^{(i)}$ defined by $U_k^{(i)} g_i=  g_i S_k^{(i)}  $, which has to exist due to condition (i). Since $S_k|\Psi_s\rangle=|\Psi_s\rangle$, we obtain then the desired state in any branch of the protocol. Iterating $LOCC_j$ steps we obtain all possible all-det-$LOCC_{\N}$ protocols. However, this process is very restricted. If party 1 has used the above POVM, then all parties different from 1 have to fulfill condition (i) in Lemma 3 for the corresponding symmetries. If an $LOCC_2$ protocol was to be implemented subsequently, new symmetries must be found that commute as well with the operators $G_i$ for $i>2$ (and also with $H_1$). This is a very strong constraint. Notice, for example, that in the case of qubits the only positive operator that commutes with two different non-commuting unitaries is the identity. Being the set of reachable states very restricted as well, it emerges as a reasonable conjecture that naturally extends the bipartite case whether all-det-$LOCC_{\N}$ protocols are sufficient for all possible LOCC transformations in the SLOCC classes we consider. Theorem 1 and Lemma 3 allow to characterize all states taking part in these transformations and this would make the investigation of $LOCC_{\N}$ protocols fully feasible. In fact, in Sec. \ref{secmes} we use Lemma 3 to determine those states which are not reachable by $LOCC_{\N}$ but convertible by all-det-$LOCC_{\N}$ protocols. We discuss therein the connection of this clearly relevant class of states with the MES of states (see below for the definition) and show moreover that entanglement measures can be estimated. In Sec. \ref{secpaulis} we construct an SLOCC class for an arbitrary number of parties with non-trivial finite stabilizer where we show that all LOCC transformations are indeed achievable by all-deterministic transformations. However, as we have shown with a particular example in \cite{short}, it turns out that this is in general not true. In Sec. \ref{secexample} we discuss in more detail constructions of transformations that can be implemented by $LOCC_{\N}$ but which require an intermediate probabilistic step. Before all that, we conclude this section by discussing an interesting effect that we outlined in \cite{short}.

\subsection{Unlocking and locking the power of other parties}\label{seclocking}

As we have discussed above, in an all-det-$LOCC_{\N}$ protocol the implementation of an $LOCC_j$ step puts strong constraints on the subsequent $LOCC_i$ steps that might be carried out by some party $i$. It can actually happen that a state that is convertible by $LOCC_j$ by any party is no longer convertible by any other party once one party acts.
We call this effect locking the power of other parties. It occurs for instance in any SLOCC class with a stabilizer such that $S^{(j)}\neq\one$ $\forall j$ whenever $S\neq\one$.
A simple example is given by considering a state $\ket{\Psi_s}$ with $S_{\Psi_s}=\{\sigma_i^{\otimes n}\}_{i=0}^3$. This state is clearly convertible by $LOCC_j$ for any party $j$ (see Lemma 3). However,
if party 1, acts first so that we obtain the state $h_1\otimes\one\ket{\Psi_s}$ such that $[H_1,S^{(1)}]\neq0$ for every non-trivial $S\in S_{\Psi_s}$, then the state is no longer convertible
by $LOCC_j$ for $j\neq1$ [see Fig. \ref{unlock} (a)]. \\ Perhaps more interestingly, the opposite effect, namely that one party can enable another party to implement a deterministic transformation, can take place as well. We termed this effect unlocking the power of the others. In this case, we consider two different parties $j$ and $k$. Then, we could have a starting state which is not $LOCC_k$-convertible, but for which party $j$ might execute an $LOCC_j$ transformation in such a way that the output state is now $LOCC_k$-convertible. We have observed that this effect exists in the L-state family (cf.\ Sec.\ \ref{seclstate}). In particular, an example is given by considering a state $g_1\otimes \one^{\otimes 3}\ket{L}$, where $\g_1=((2p-1) x,(2p-1) x,x)$. The parameters $x $ and $p$ are chosen such that the following two conditions are fulfilled: (i) $[G_1,S^{(1)}]\neq 0$, $\forall S \in S_L$ ($S\neq \one^{\otimes 4}$) and (ii) $H_1=\one/2+x\sum_{j=1}^3 \sigma_j$ is a positive operator. Due to Lemma 3 condition (i) implies that the state $g_1\otimes \one^{\otimes 3}\ket{L}$ cannot be converted by any $LOCC_k$ with $k>1$. Condition (ii) ensures that $G_1$ is positive and that the operators $\sqrt{p}h_1 g_1^{-1}, \sqrt{1-p}h_1 \sigma_3 g_1^{-1}$ where $h_1= \sqrt{H_1}$ form a valid measurement.
If party 1 implements this measurement and all other parties apply for the first (second) measurement outcome the identity ($\sigma_3$), respectively, the state $g_1\otimes \one^{\otimes 3}\ket{L}$ is transformed
deterministically into the state $h_1\otimes \one^{\otimes 3}\ket{L}$. Note that $H_1$ (and $h_1$) commute now with $U$. Hence, once this $LOCC_1$ protocol is applied, any of the other parties can apply a LOCC using
the symmetry $U$ [see Fig. \ref{unlock} (b)].
A similar effect is presented in \cite{GourSpek06} in order to show that entanglement of assistance is not an entanglement measure.\\
\begin{figure}
\centering
\includegraphics[width=0.4\textwidth]{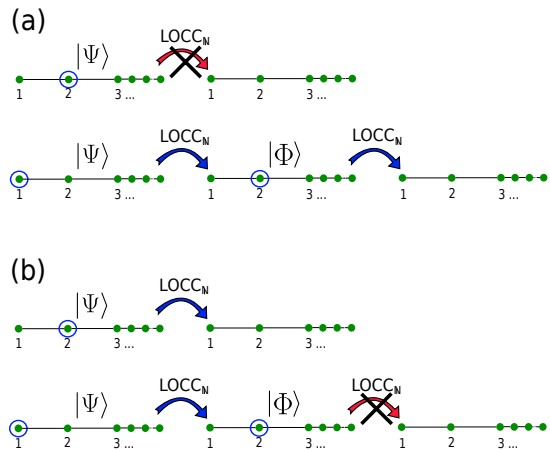}
\caption{(a) Locking the power of others: After party 1 measured, party 2 cannot apply any nontrivial deterministic transformation even though party 2 could have performed one before.
(b) Unlocking the power of others: Only after party 1 has measured can party 2 apply a nontrivial deterministic transformation.}
\label{unlock}
\end{figure}
In contrast to locking the power of others the effect that an $LOCC_j$ protocol can transform a state which is not $LOCC_k$-convertible into a $LOCC_k$-convertible one can only occur if the stabilizer of the considered SLOCC class has a particular structure. The following lemma identifies classes where it is certainly not possible to unlock any power. For example, generic 4-qubit states and generic 3-qutrit states fall into this category.

\begin{lemma}
\label{lemma_specialsym}
If $\ket{\Psi_s}$ is such that $S_i S_j \propto S_j S_i$ $\forall S_i,S_j\in S_\Psi$, then unlocking the power of others is impossible. \end{lemma}

\begin{proof}
For this to be true, it suffices to prove that there does not exist a positive operator $H$ such that $[H,S]=0$ and $[\sum_i p_i S_i^\dagger H S_i,S]\neq 0$, for any probability distribution $p_i$. Using that $S_i S =e^{i\phi} S S_i$ \footnote{Note that the proportionality factor must be a phase as all operators here are unitary.} it is straightforward to see that $[H,S]=0$ implies that $[\sum_i p_i S_i^\dagger H S_i,S]= 0$ for any $p_i$. Hence, any symmetry which can be used by the other parties after this transformation can already be used before this transformation. \end{proof}

\section{Maximally entangled set under $LOCC_{\N}$ and estimation of entanglement measures}\label{secmes}

An interesting consequence of the characterization of LOCC transformations among bipartite pure states \cite{nielsen} is that it firmly establishes the state in $\C^d\otimes\C^d$
\begin{equation}\label{maxent}
|\phi_d^+\rangle=\frac{1}{\sqrt{d}}\sum_{i=1}^d|ii\rangle
\end{equation}
as the maximally entangled state. This is because no state in $\C^d\otimes\C^d$ can be transformed to it while it can be transformed to any other state in $\C^d\otimes\C^d$. Thus, from the point of view of entanglement
theory, it is not surprising that the state $|\phi_d^+\rangle$ is the most useful in all bipartite applications such as teleportation. It is then natural to look for states with similar properties in the multipartite regime.
However, in general there exists no state with the above stated property when one considers $\C^{d_1}\otimes \ldots \otimes \C^{d_n}$ with $n>2$. Therefore, in \cite{dVSp13}
some of us have introduced the notion of the maximally entangled set (MES) in $\C^{d_1}\otimes \ldots \otimes \C^{d_n}$. This is the minimal set of states that allows to reach by LOCC any other truly $n$-partite entangled pure state in
that Hilbert space (or in a restricted SLOCC class). Thus, all states in the MES are not reachable by LOCC (from states of the same dimensionality and number of parties) and any reachable state can be
obtained by at least one state in the MES. Stated differently (but perhaps less intuitively) the MES is the set of truly $n$-partite entangled pure states that are not reachable by any other state.
Investigations on 4-qubit and 3-qutrit states showed that in these cases the MES is, in contrast to the 3-qubit case, of full-measure \cite{dVSp13,SpdV16,HeSp15}. This is because almost all states are isolated,
i.e.\ the states are indeed not reachable but they are also not convertible by LOCC. Isolated states are then useless in this context and the most interesting states correspond to convertible states in the MES.
Note that, our investigations on $LOCC_{\N}$ do not allow us to identify the MES because non-reachable states in this scenario might be reachable by infinite-round LOCC. Nevertheless, the full MES must be a subset
(not necessarily strict) of the MES under $LOCC_{\N}$. Notice that this latter set is precisely characterized as those states that do not fulfill the conditions of Theorem 1. As explained above, we would like to
identify convertible rather than isolated states. It would therefore be interesting to see which of the aforementioned states have this property. Using Lemma 3 one can easily characterize the states in the MES under
$LOCC_{\N}$ which are convertible under all-det-$LOCC_{\N}$ protocols.

\begin{theorem}\label{mes}
A state $g\ket{\Psi_s}$ is in the MES under the restriction to $LOCC_{\N}$ and is convertible via all-det-$LOCC_{\N}$ protocols iff
\bi \item[(i)] $\exists S\in S_{\Psi_s}$ such that $[G,S]=0$ ($S\neq\one$) and
\item[(ii)] $\forall S\in S_{\Psi_s}$ such that $[G_i,S^{(i)}]=0$ $\forall i\neq j$ it holds that $[G_j,S^{(j)}]=0$.\ei
\end{theorem}
\begin{proof}
Let us first show that these conditions are necessary. Since the states must be convertible by all-det-$LOCC_{\N}$ (and hence by $LOCC_j$ for some party $j$) but not reachable, they must fulfill the conditions of Lemma 3 and violate at least one of the conditions of Theorem 1. However, satisfying non-trivially condition (i) in Lemma 3 implies that condition (i) in Theorem 1 also holds for a symmetry which is not the identity. Therefore, the states must violate condition (ii) in Theorem 1. The conditions given in this theorem express this fact. It thus remains to see that the conditions are sufficient, i.e.\ whenever they hold, condition (ii) in Lemma 3 is satisfied non-trivially. For that, consider the symmetry $S$ of condition (i) here, which must be different from the identity at least for one party. Take without loss of generality that $S^{(1)}\neq\one$ and consider the twirling operation $T$ under the group generated by $\one$ and $S^{(1)}$ (cf.\ Eq.\ (\ref{twirling})). We then have that $T(G_1)=G_1$. Define now $H=G_1+X-T(X)$ for a Hermitian operator $X$ such that $[X,S^{(1)}]\neq0$ (so that $[H,S^{(1)}]\neq0$). With this $T(H)=G_1$, which allows to fulfill condition (ii) in Lemma 3 non-trivially. It only remains to see that some $X$ can be chosen so that $H$ is a positive operator. However, this is clear from the fact that the set of positive operators is open. We can certainly choose $X$ such that $0<||X-T(X)||<\epsilon$. This implies that $||G_1-H||<\epsilon$ and by choosing $\epsilon$ small enough the positiveness of $G_1$ guarantees that of $H$.
\end{proof}
It is interesting to observe that the sufficiency part of the above proof works no matter which party $j$ is to perform the measurement within the $LOCC_j$ step as long as the symmetry of condition (i) obeys $S^{(j)}\neq\one$. Thus, although it is in principle necessary for the states in the MES under $LOCC_{\N}$ that are convertible by all-det-$LOCC_{\N}$ to be convertible by $LOCC_j$ for just one party $j$, it turns out that they are convertible by $LOCC_j$ for \emph{any} party $j$ as long as there exists a symmetry fulfilling (i) such that $S^{(j)}\neq\one$. In fact, if $S^{(j)}\neq\one$ $\forall j$ and $\forall S\in S_\Psi$ ($S\neq\one$), any symmetry which can be used by one party can be used by any other to perform a non-trivial transformation.

The investigation of LOCC protocols also plays a crucial role for the quantification of entanglement. Entanglement measures must be quantities that do not increase under LOCC transformations. Functionals that obey this rule can be constructed on pure mathematical grounds. However, they are often very hard to compute and/or to be provided with a clear operational meaning. In order to close this gap, in \cite{ScSa15,SaSc15} we introduced entanglement measures with a very clear operational meaning in terms of usefulness under LOCC manipulation and that can be computed whenever the possible LOCC transformations are characterized: the source and accessible entanglement. Given any pure state $\ket{\Psi}$ \footnote{In this paper we only consider pure states but the entanglement measures can be equally defined for mixed states.}, we define its accessible set, $M_a(\ket{\Psi})$, as the set of states which are reachable by LOCC from $\ket{\Psi}$. Analogously, we define its source set, $M_s(\ket{\Psi})$ as the set of states which can reach $\ket{\Psi}$ by LOCC. By considering some convenient measure $\mu$ in the set of LU-equivalence classes, we can then obtain the source volume $V_s(\ket{\Psi})=\mu[M_s(\ket{\Psi})]$ and the accessible volume $V_a(\ket{\Psi})=\mu[M_a(\ket{\Psi})]$ of a given state. The intuition behind these quantities is that the larger the accessible volume of a state is, the relatively more useful the state is and viceversa for the source volume. With this, the accessible entanglement and the source entanglement are defined by
\begin{equation}\label{eas}
E_{a}(\ket{\Psi})=\frac{V_a(\ket{\Psi})}{V_a^{sup}},\quad E_{s}(\ket{\Psi})=1-\frac{V_s(\ket{\Psi})}{V_s^{sup}},
\end{equation}
where $V_a^{sup}$ ($V_s^{sup}$) denote the supremum of the accessible (source) volume according to the measure $\mu$ so that the measures are normalized between 0 and 1 \footnote{It should be noted here that $E_a$ and $E_s$ are valid entanglement measures for any choice of measure $\mu$.}. Notice that, by construction, these quantities are entanglement measures, i.e.\ they do not increase under LOCC transformations. Note that these ideas can be used to define not only two entanglement measures, but two classes of entanglement measures \cite{ScSa15,SaSc15}. To this end one does not only consider transformations among states within the same Hilbert space, but for instance measures the amount of $n+1$--qubit states which can be transformed via LOCC into the $n$--qubit state of interest. Due to their operational meaning, all these measures are easily shown to be entanglement measures. Moreover, even if we consider in the following only the measures as described above, these results can also be applied to any other measure in these classes.

Since these quantities measure the relative usefulness of states, its computation can be restricted to a given SLOCC class (i.e.\ one can consider a measure $\mu$ supported only on a particular SLOCC class) \cite{ScSa15,SaSc15}. Nevertheless, its computation obviously relies on characterizing all possible LOCC transformations under the class of states of interest. Here, we have considered the particular class of LOCC protocols that can be implemented with finitely many rounds of classical communication $LOCC_{\N}$. On the analogy to the above concept we can introduce the accessible and source set, volume and entanglement under $LOCC_{\N}$ for a given state $\ket{\Psi}$: $M_{a,s}^{LOCC_{\N}}(\ket{\Psi})$, $V_{a,s}^{LOCC_{\N}}(\ket{\Psi})$ and $E_{a,s}^{LOCC_{\N}}(\ket{\Psi})$. Obviously then, $E_a^{LOCC_{\N}}(\ket{\Psi})$ and $E_s^{LOCC_{\N}}(\ket{\Psi})$ do not increase under $LOCC_{\N}$ protocols and quantify the usefulness of states for protocols in this setting. However, it should be stressed that these quantities are not necessarily entanglement measures. This is because they might increase under LOCC if it turns out that there exist LOCC transformations among pure states which are not implementable with finitely many rounds of communication. To see this, take any $LOCC_{\N}$ non-related states $\ket{\Psi}$ and $\ket{\Phi}$ such that $E_a^{LOCC_{\N}}(\ket{\Psi})<E_a^{LOCC_{\N}}(\ket{\Phi})$ and suppose that $\ket{\Psi}$ could be converted to $\ket{\Phi}$ by an infinite-round LOCC protocol. A similar argument applies to $E_s^{LOCC_{\N}}$. However, interestingly, these quantities can be used to bound the source and accessible volumes. Moreover, one can define analogous functionals under the constraint to all-det-$LOCC_{\N}$ protocols. The inclusion all-det-$LOCC_{\N}$ $\subset LOCC_{\N}\subset$ LOCC immediately yields
\begin{align}
V_a(\ket{\Psi})&\geq V_a^{LOCC_{\N}}(\ket{\Psi})\geq V_a^{all-det}(\ket{\Psi}),\nonumber\\
V_s(\ket{\Psi})&\geq V_s^{LOCC_{\N}}(\ket{\Psi})\geq V_s^{all-det}(\ket{\Psi}).
\end{align}

Thus, using the insights of Lemma 3 one can place non-trivial estimates for the source and accessible volumes via all-det-$LOCC_{\N}$ protocols. In particular, given any state $g|\Psi_s\rangle$ fulfilling the premises of Lemma 3 we see that any other state $h_1\otimes g_2 \otimes \ldots \otimes g_n\ket{\Psi}$ can be reached by $LOCC_1$ whenever there exists a probability distribution $\{p_k\}$ such that
\begin{equation}
G_1=\sum_{k} p_k (S_k^{(1)})^\dagger H_1 S_k^{(1)}
\end{equation}
for all those symmetries $\{S_k\}$ such that $[G_i,S_k^{(i)}]=0$ $\forall i>1$. Once the symmetries and states are specified it is then not difficult to determine all such states $h_1\otimes g_2 \otimes \ldots \otimes g_n\ket{\Psi}$. The process is then iterated to consider $LOCC_j$ protocols by other parties (if possible). If we start from a MES state as characterized in Theorem \ref{mes}, one obtains all possible all-det-$LOCC_{\N}$ protocols and can then determine $V_a^{all-det}$ and $V_s^{all-det}$ for any state obtainable via these operations, i.e. any state convertible and/or reachable via all-det-$LOCC_{\N}$ in the corresponding SLOCC class.

Note that in case $V_a^{sup}$ or $V_s^{sup}$ are known, one can easily bound the corresponding (normalized) entanglement measures $E_a$ (lower bound) and $E_s$ (upper bound). In particular, $V_s^{sup}$ is always given by the whole volume, as it is the source volume of a product state to which any state can be transformed to. In case these quantities are unknown, one can nevertheless compare states with each other using the unnormalized quantities $V_a$ and $V_s$.

\section{Examples of SLOCC classes where all-det-$LOCC_{\N}$ protocols are sufficient}\label{secpaulis}

In this section we show that for any SLOCC class of $n$-qubit states, whose representative has the stabilizer $\{\sigma_i^{\otimes n}\}_{i=0}^3$ it holds that $LOCC_{\N}$ coincides with all-det-$LOCC_{\N}$.
In fact, it can be easily shown that for the here considered SLOCC classes we have
\bea SEP=LOCC=LOCC_{\N}=all-det-LOCC_{\N}\eea
for pure state transformations.

Given the fact that the LOCC transformations can be easily characterized in this case, this implies first that the entanglement measures introduced in \cite{ScSa15} and discussed in the previous section can be easily computed.
Moreover, it becomes easier to introduce new entanglement measures for pure states for these SLOCC classes, as any function which is non-increasing under all-det-$LOCC_{\N}$ is a valid entanglement measure (for these SLOCC classes). That is, it is only necessary to show that the function is non-increasing under $LOCC_j$ for any party $j$. In particular, the classical communication among the parties does not need to be taken into account,
which simplifies the investigation enormously \footnote{Note however, that in contrast to entanglement monotones, where the average entanglement is non-increasing under LOCC, one needs to consider here only those unilocal $LOCC$ operations,
which lead deterministically to the final state, i.e. $LOCC_j$.}. It should be noted here that there exists no bipartite or three qubit state with only these symmetries. However, we are going to construct here explicit examples of $2^m$-qubit states ($m\geq2$) whose stabilizer we prove to be $\{\sigma_i^{\otimes 2^m}\}_{i=0}^3$.

Let us now show that for any SLOCC class, whose representative, $\ket{\Psi_s}$, has the stabilizer $\{\sigma_i^{\otimes n}\}_{i=0}^3$ it holds that $SEP = all-det-LOCC_{\N}$. In order to do so, we use the characterization of all separable transformations
among states in the 4-qubit generic SLOCC classes \cite{dVSp13,SaSc15} (whose stabilizer is $\{\sigma_i^{\otimes 4}\}_{i=0}^3$). It has been shown in \cite{dVSp13} that the only states which are reachable via SEP are
(up to permutations of the particles) of the form (i) $h^1 \otimes h^2_w \otimes h^3_w\otimes h^4_w \ket{\Psi}$, with $w\in\{x,y,z\}$ and $h^1 \neq h^1_w$ or (ii) $h^1 \otimes \one^{\otimes 3}\ket{\Psi}$ with $h^1 \not \propto \one$ arbitrary. Here, $h_w \propto \one + \alpha \sigma_w$ with $\alpha \in \R$.
Notice that this coincides with the set of reachable states by all-det-$LOCC_{\N}$ (and $LOCC_{\N}$) as given by Theorem 1 here using that in this case $S_{\Psi_s}=\{\sigma_i^{\otimes 4}\}_{i=0}^3$.
This result has been derived by considering Eq. (\ref{condsep}) and noticing that this equation can only hold if all but one operator $H_i$ commute with at least one symmetry.
Moreover, in \cite{SaSc15} we characterized all possible separable transformations among states in a generic SLOCC class of four qubits. Furthermore, we showed there that all these transformations can be realized via a all-det-$LOCC_{\N}$ protocol
by explicitly presenting the protocol. The characterization of the possible separable transformations as well as the corresponding all-det-$LOCC_{\N}$ transformations can be straightforwardly generalized to an arbitrary number of qubits. The reason is that if Eq.\ (\ref{condsep}) has to hold now for $n>4$ qubits, by tracing out here any $n-4$ parties we recover the same condition as in the case $n=4$ for any group of four parties. Thus, the convertible and reachable states are such that at least three local operators must commute with the same $\sigma_i$ for any group of four parties. Hence, all theses results generalize
to the cases where more than four qubits are considered. As a result we have that for the here considered SLOCC classes any separable pure state transformations can be realized with an all-det-$LOCC_{\N}$ transformation and therefore, in particular any
LOCC and $LOCC_{\N}$ transformation can be realized with such a simple protocol.\\

Let us now present a class of $2^m$-qubit states, with $m\geq 2$, whose stabilizer we show to be $\{\sigma_i^{\otimes {2^m}}\}_{i=0}^3$ in Appendix \ref{lemma6}. In order to do so we use the following notation. We call a four-qubit seed state of the form
\begin{align}
 \ket{\psi_2(\vec{\alpha})} = \sum_{i=0}^3 \alpha_i \sigma_i^{(2)} \sigma_i^{(4)} \ket{\phi^+}_{12}\ket{\phi^+}_{34}, \label{eq:DefGen4qb}
\end{align}
generic, if the normalized complex vector $\vec{\alpha} = (\alpha_0,\alpha_1,\alpha_2,\alpha_3) \in \C^4$ is such that the stabilizer of the state is $S_{\psi_2(\vec{\alpha})} = \{\sigma_i^{\otimes 4}\}$.
As shown in \cite{SpdV16} this is the case if $\alpha_i^2 \neq \alpha_j^2$ for $i \neq j$ and if there exists no $q \in \mathbb{C}\setminus\{1\}$ that yields
$\{\alpha_i^2\}_{i=0}^3 = \{q \alpha_i^2\}_{i=0}^3$. We label the qubits of
a $n$-qubit state by indices $0,\ldots,n-1$ and denote by $U_{i,j}$ the operator that permutes subsystems $i$ and $j$.
Using this notation we state the following lemma.

\begin{lemma}
\label{lem:2mstates}
 For almost all normalized $\vec{\alpha} \in \C^4$ the stabilizer of the recursively defined $2^m$-qubit state
 \begin{align*}
 \ket{\psi_m(\vec{\alpha})} \propto \frac{1}{2}(\one + U_{0,2^{m-1}}) \ket{\psi_{m-1}(\vec{\alpha})}^{\otimes 2}, \ m > 2,
\end{align*}
 is $S_{\psi_m(\vec{\alpha})} = \{\sigma_i^{\otimes 2^m}\}_{i=0}^3$.
\end{lemma}

We suspend the proof of this lemma to the appendix for the sake of readability \footnote{Note that the generic set of four-dimensional vectors considered here slightly differs from the one above (for details see Appendix \ref{lemma6}).}.
It is however worthwhile to look at how $\ket{\psi_m(\vec{\alpha})}$ can be created from states of fewer qubits (see also Fig. \ref{fig:1}) \footnote{Note that these states are indeed multipartite entangled as they possess non--trivial symmetries and as any product state, $\ket{\Psi}\ket{\Phi}$, must have a symmetry of the from $\one \otimes S$, where $S$ denotes a non--trivial symmetry of $\ket{\Phi}$.}. The proof of Lemma \ref{lem:2mstates} fundamentally depends
on this construction.\\

If we want to construct the state $\ket{\psi_m(\vec{\alpha})}$ we can start by entangling two copies of $\ket{\psi_2(\vec{\alpha})}$ via the operator
$1/2 (\one+U_{0,4})$, which projects onto the symmetric subspace of particles zero and four, to obtain $\ket{\psi_3(\vec{\alpha})} \propto 1/2 (\one+U_{0,4}) \ket{\psi_2(\vec{\alpha})}^{\otimes 2}$. Then, two copies of the
resulting state are entangled via $1/2 (\one+U_{0,8})$ to obtain the 16-qubit state
\begin{align}
&\ket{\psi_4(\vec{\alpha})} \propto 1/2 (\one +U_{0,8})\ket{\psi_3(\vec{\alpha})}^{\otimes 2}\\
&\propto 1/8 (\one +U_{0,8})(\one +U_{8,12}) (\one+U_{0,4})\ket{\psi_2(\vec{\alpha})}^{\otimes 4}. \label{eq:ent16qb}
\end{align}
That is, the 16-qubit state can also be created by entangling four copies of $\ket{\psi_2(\vec{\alpha})}$ with operators $ 1/2(\one+U_{i,j})$ given in Eq. (\ref{eq:ent16qb}).
Continuing with this procedure we eventually obtain the desired $2^m$-qubit state and see that it can also be expressed as
\begin{align*}
 \ket{\psi_m(\vec{\alpha})} \propto K(m,2) \ket{\psi_2(\vec{\alpha})}^{\otimes 2^{m-2}}.
\end{align*}
Here, $K(m,2)$ denotes the operator that includes all entangling operators of the form $1/2 (\one+U_{i,j})$ that are needed to generate $\ket{\psi_m(\vec{\alpha})}$
from $2^{m-2}$ copies of $\ket{\psi_2(\vec{\alpha})}$ (see Fig. \ref{fig:1}), e.g. $K(4,2) = 1/8 (\one +U_{0,8})(\one +U_{8,12}) (\one+U_{0,4})$. Analogously, we could have
started with $2^{m-k}$ copies of the state $\ket{\psi_k(\vec{\alpha})}$, with $2 \leq k < m$, as elementary building blocks and entangled them to obtain the desired state.
That is, we can express the final state
as
\begin{align}
  \ket{\psi_m(\vec{\alpha})} \propto K(m,k) \ket{\psi_k(\vec{\alpha})}^{\otimes 2^{m-k}}, \label{eq:KtoM}
\end{align}
e.g. $K(4,3) = 1/2(\one + U_{0,8})$.
This is also illustrated in Fig. \ref{fig:1}. It is easy to see that $K(m,k)$ acts nontrivially only on qubits with index
$2^k\cdot l$, for $l \in \{0,\ldots,2^{m-k}-1\}$ and that
$K(m,m) = \one$ and $K(m,m-1) = 1/2(\one+U_{0,2^{m-1}})$ hold. Due to the definition it is moreover clear
 that the recursion $K(m,k) = 1/2 (\one + U_{0,2^{m-1}})(K(m-1,k) \otimes K(m-1,k))$ holds. One may verify that
 $K(m,k)$ is the product of $2^{m-k}-1$ projectors onto the symmetric subspace of two qubits (see also Fig. \ref{fig:1}).\\

\begin{figure}
 \includegraphics[width=0.45\textwidth]{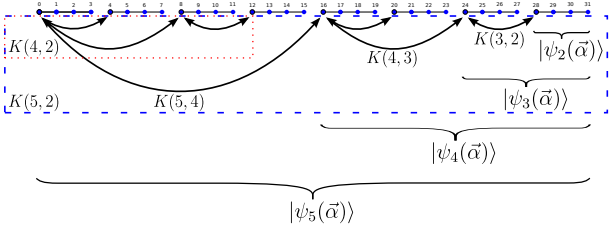}
 \caption{This figure depicts the structure of the operators $K(m,k)$ used to create the 32-qubit state $\ket{\psi_5(\vec{\alpha})}$ from eight
 copies of the four-qubit state $\ket{\psi_2(\vec{\alpha})}$. Each operator of the form $1/2(\one + U_{i,j})$ is represented by a bow connecting qubits $i$ and $j$. With the help of this figure, it is easy to see that the
 operator $K(4,2) = 1/8 (\one +U_{0,8})(\one +U_{8,12}) (\one+U_{0,4})$ is the product of all operators $1/2(\one + U_{i,j})$ that appear inside the region with dotted boundary. Note that the nesting of the bows corresponds to the order in which the operators associated to these
 bows have to be applied to the copies of $\ket{\psi_2(\vec{\alpha})}$, e.g. the operators corresponding to the innermost bows have to be applied first. An analogous statement holds for the operator $K(5,2)$, which is the product of all operators $1/2(\one + U_{i,j})$ that appear in the region with
 dashed boundary. One easily sees that in each copy of $\ket{\psi_2(\vec{\alpha})}$ the first qubit is singled out as it is the only
 one on which a non-trivial operation acts.}
  \label{fig:1}
\end{figure}

As shown before, all $LOCC_{\N}$ transformations in the SLOCC class of a quantum state $\ket{\psi_m(\vec{\alpha})}$ as
described in Lemma \ref{lem:2mstates} can be accomplished via deterministic steps. In combination with
our characterization of all-det-$LOCC_{\N}$ reachable and convertible states we moreover obtain a characterization of all $LOCC_{\N}$ transformations in these classes. In the next section we show, in contrast to these SLOCC classes, that there exist SLOCC classes for which certain $LOCC_{\N}$
transformations require a probabilistic step.

\section{all-det-$LOCC_{\N}$ protocols are not sufficient for $LOCC_{\N}$ pure-state transformations}\label{secexample}

Leaving aside the bipartite and the three-qubit case, all previous investigations on the subject and the results obtained here indicate that the possibility to transform a pure state by deterministic LOCC into an LU-inequivalent one is very rare. The conditions required on the states for such a protocol to be possible are so stringent that they are almost never met. However, when this is indeed the case it turns out in all cases studied so far that all these constraints only leave room to relatively simple LOCC protocols. To our knowledge, all instances known so far of LOCC transformations in the multipartite case (including the bipartite setting) can always be realized by all-det-$LOCC_{\N}$ protocols. This might had lead to conjecture that this is always the case. However, we have considered a particular example in \cite{short} that shows that this is not the case. There exist $LOCC_{\N}$ transformations among pure states that require more elaborate LOCC protocols which include non-deterministic intermediate steps. On the one hand, this shows that the rich structure of LOCC maps can be exploited to devise more sophisticated protocols that allow for otherwise impossible transformations. On the other hand, this leaves little hope to have a general characterization of $LOCC_{\N}$ convertible states. Although in many SLOCC classes all-det-$LOCC_{\N}$ transformations are the only possible protocols for pure-state transformation and convertible states are therefore characterized by Lemma 3, there exist other families in which this is not the case.

In this section we discuss a general construction that allows to obtain such examples and that we used to provide the particular instance presented in \cite{short}. Not surprisingly, the SLOCC class used to find this non-conventional behaviour is again the L-state family, which we have considered in Sec. \ref{seclstate}. We focus on states of the form $g_1\otimes g_2\otimes\one\otimes\one|L\rangle$ and hence denote them by $\{g_1,g_2\}$, $\{G_1,G_2\}$ or $\{\g_1,\g_2\}$. In order to obtain the desired examples, we provide a protocol that implements the transformation $\{\g_1,\g_2\}\to\{\h_1,\h_2\}$ by choosing an initial state such that
\begin{equation}\label{initialg1}
[G_i,S^{(i)}_k]\neq0,\quad \forall S_k\in S_L\, (S_k\neq\one)\textrm{ and } i=1,2.
\end{equation}
Thus, Lemma 3 guarantees that the initial state cannot be converted by an $LOCC_j$ protocol $\forall j$. Hence, $\{\g_1,\g_2\}$ is not convertible to any other state by all-det-$LOCC_{\N}$ and any deterministic transformation starting from this state necessarily requires intermediate non-deterministic steps. The protocol has two steps. First, party 1 implements a two-outcome POVM that leads to the intermediate states $\{h_1,g_2\}$ and $\{h_1\sigma_3,g_2\}$. The lack of symmetry of $G_2$ guarantees that the two outputs are LU-inequivalent. The idea now is that both $H_1$ and $\sigma_3H_1\sigma_3$ fulfill the premises of Lemma 3 so that an unlocking effect takes place and each intermediate state can now be transformed by $LOCC_2$ into the same state $\{h_1,h_2\}$. For this to be possible we choose
\begin{equation}\label{intermediateh1}
[H_1,U]=[H_1,U^2]=0.
\end{equation}
This means that all entries of $\h_1$ are equal, i.e. $\h_1=(x,x,x)$, for some $x\in \R$. In order to achieve the desired action, the initial POVM is taken with elements $M_1=\sqrt{p}h_1g_1^{-1}$ and $M_2=\sqrt{1-p}h_1\sigma_3g_1^{-1}$. This is a valid measurement only if $M_1^\dag M_1+M_2^\dag M_2=\one$ which amounts to
\begin{equation}
\g_1
    =x\left(\begin{array}{c}
         2p-1 \\
         2p-1 \\
         1 \\
       \end{array}
     \right).
\end{equation}
It should be clear then that both conditions (\ref{initialg1}) and (\ref{intermediateh1}) can be met by choosing properly the parameters ($\g_1$ must have its first two components equal). It only remains to find $LOCC_2$ transformations to map both $\{h_1,g_2\}$ and $\{h_1\sigma_3,g_2\}\simeq_{LU}\{h_1,g_2\sigma_3\}$ into the same final state $\{h_1,h_2\}$. For this, party 2 in each case can use a POVM with at most three elements:
\begin{equation}
M_1,M_2,M_3\propto h_2g_2^{-1},h_2Ug_2^{-1},h_2U^\dag g_2^{-1}
\end{equation}
for the branch corresponding to $\{h_1,g_2\}$ and
\begin{equation}
\tilde{M}_1,\tilde{M}_2,\tilde{M}_3\propto h_2\sigma_3g_2^{-1},h_2U\sigma_3g_2^{-1},h_2U^\dag\sigma_3 g_2^{-1}
\end{equation}
for the branch corresponding to $\{h_1,g_2\sigma_3\}$. Here we have used that $U^2=-U^\dag$ and that the used symmetry has to commute with $H_1$ now. Normalization requires then probability distributions $\{q_i\}$ and $\{\tilde{q}_i\}$ such that
\begin{align}
q_1H_2+q_2U^\dag H_2U+q_3UH_2U^\dag&=G_2,\nonumber\\
\tilde{q}_1H_2+\tilde{q}_2U^\dag H_2U+\tilde{q}_3UH_2U^\dag&=\sigma_3G_2\sigma_3.
\end{align}
If we denote $\g_2=(g_x,g_y,g_z)$ and $\h_2=(h_x,h_y,h_z)$, this amounts to both $(g_x,g_y,g_z)$ and $(-g_x,-g_y,g_z)$ belonging to the polytope generated by the convex hull of the points $(h_x,h_y,h_z)$, $(h_z,h_x,h_y)$ and $(h_y,h_z,h_x)$. It is straightforward to verify that this polytope is contained in the plane $x+y+z=h_x+h_y+h_z$. Finally, if we choose then $h_x+h_y+h_z=0$ and the origin is contained in the polytope, it will then include points both of the form $(g_x,-g_x,0)$ and $(-g_x,g_x,0)$ achieving the desired condition.

The particular instance of the above construction we have used to give the example of \cite{short} is the transformation
\begin{equation}
\left\{\left(
         \begin{array}{c}
           x \\
           x \\
           2x \\
         \end{array}
       \right),\left(
         \begin{array}{c}
           x \\
           -x \\
           0 \\
         \end{array}
       \right)\right\}\to\left\{\left(
         \begin{array}{c}
           2x \\
           2x \\
           2x \\
         \end{array}
       \right),\left(
         \begin{array}{c}
           x \\
           x \\
           -2x \\
         \end{array}
       \right)\right\},
\end{equation}
where $x\neq0$ is any real number small enough such that all the above vectors lead to positive operators (i.e.\ their norm is strictly smaller than 1/2). This transformation can be implemented by LOCC following the protocol explained above by choosing $p=3/4$, $q=(1/3,0,2/3)$ and $\tilde{q}=(1/3,2/3,0)$ (for any allowed value of $x$). However, it should be clear from the above discussion that other instances can be found.

\section{Conclusions}
Together with the companion paper \cite{short}, we have considered transformations among pure n-partite states which can be realized with $LOCC_{\N}$. Considering a large set of SLOCC classes we have provided a very simple
characterization of all reachable states by $LOCC_{\N}$ and have characterized as well the convertible states under all-det-$LOCC_{\N}$ protocols. Due to our characterization of all-det-$LOCC_{\N}$ convertibility and that previously known LOCC transformations fall, up to our knowledge, into this category which is a natural generalization of bipartite transformations, we addressed the questions whether all $LOCC_{\N}$ are realizable via a protocol which is deterministic in each step. We show that this is not the case by providing the general tools to construct examples of pairs of states, where the initial state can only be transformed into the final state if a probabilistic step is involved. Moreover, we identified SLOCC classes of arbitrary many qubits for which any separable pure state transformation can be realized by LOCC via an all-det-$LOCC_{\N}$ protocol and provided explicit examples of $2^m$-qubit states belonging to these SLOCC classes.

Taking into account the results on LOCC transformations (including infinitely many rounds) the following picture emerges. Due to the strong constraints on possible pure state transformations via LOCC the convertibility properties of multipartite states can be characterized to a large  extent of generality. These constraints are already imposed by convertibility under the larger class of separable operations \cite{Gour}, even though not all separable pure state transformations can be realized via LOCC \cite{HeSp15}. This highlights one of the differences between bipartite and multipartite entanglement manipulation. As we show here, considering a physically meaningful subset of LOCC, namely $LOCC_{\N}$, reveals again a very clear difference between the bipartite and multipartite case. This is due to the fact that considering a reasonable generalization of protocols realizing bipartite state transformations, that allows to focus on relatively simple protocols (all-det-$LOCC_{\N}$), is not sufficient to capture all possible state transformations in the multipartite case. Hence, while in the bipartite case we have that all-det-$LOCC_{\N} = LOCC_{\N}=LOCC=SEP$, in the multipartite case we have all-det-$LOCC_{\N}\subsetneq LOCC_{\N}$ and $LOCC \subsetneq SEP$. It remains open for future study whether infinite-round LOCC can provide an advantage over $LOCC_{\N}$, i.e. whether $LOCC_{\N}=LOCC$ holds for pure state transformations. Furthermore, the consideration of the differences and similarities between maps which can/cannot be implemented via all-det-$LOCC_{\N}$ and the multicopy case would be very interesting. Moreover, the detailed investigation of the power of states which play a dominant role in quantum information theory and/or condensed matter physics, such as matrix product states \cite{mps} and projected entangled pairs states \cite{VeCi04} would be very intriguing.

\begin{acknowledgments}
We thank Gilad Gour and Nolan Wallach for pointing out and explaining the proof presented in Ref. \cite{bookWallach} to us.
After completing this manuscript it was proven that generic n-qubit-states ($n\geq 5$) have actually a trivial stabilizer \cite{GoKr16}. This gives an immediate alternative proof of our Corollary 2 in this case, which is moreover extendable to infinite-round protocols.
The research of CS, DS and BK was funded by the Austrian Science Fund (FWF): Y535-N16. DS and BK thank in addition the DK-ALM: W1259-N27. The research of JIdV was funded by the Spanish MINECO through grants MTM2014-54692-P and MTM2014-54240-P and by the Comunidad de Madrid through grant QUITEMAD+CM S2013/ICE-2801.
\end{acknowledgments}

\begin{appendix}
\section{Proof of Lemma \ref{lemmaconv}\label{lemma3}}
In this appendix we present the details of the proof of  Lemma \ref{lemmaconv}, i.e. we show which states are convertible via $LOCC_1$. In order to improve readability we repeat the lemma here.\\ \\
\noindent \textit{{\bf Lemma \ref{lemmaconv}.}
A state $\ket{\Psi}\propto g\ket{\Psi_s}$ is convertible via $LOCC_1$ iff there exist $m$ symmetries $S_k\in S_{\Psi_s}$, with $m> 1$ and $H\in {\cal B} ({\cal H}_1)$, $H >0$ and $p_k> 0$ with $\sum_{k=1}^m p_k=1$ such that the following conditions hold
\bi \item[(i)] $[G_i,S_k^{(i)}]=0$ $\forall i>1$ and $\forall k \in \{1,\ldots,m\}$
\item[(ii)] $G_1=\sum_{k=1}^m p_k (S_k^{(1)})^\dagger H S_k^{(1)}$ and $H \neq S^{(1)} G_1 {S^{(1)}}^\dagger$ for any $S \in S_{\Psi_s}$ fulfilling (i).
\ei}

\begin{proof} Let us first show that the conditions are necessary. Let $\{A_k\}$ denote the measurement operators of the POVM applied by party 1 \footnote{We only need to consider those $A_k$ which are not LU-equivalent to each other, i.e. $A_j^\dagger A_j \not\propto A_k^\dagger A_k$, $\forall 1\leq j,k \leq m$.}. Denoting by $\ket{\Phi}\propto h\ket{\Psi_s}$ the output state, we have that
\begin{equation}\label{Eqconv1}
\frac{1}{\sqrt{p_k}}A_k g_1 \bigotimes_{i\neq1} g_i\ket{\Psi_s}\simeq_{LU} \frac{n_\Psi}{n_\Phi}h\ket{\Psi_s},
\end{equation}
where $n_\Psi$ ($n_\Phi$) denotes the norm of $g\ket{\Psi_s}$ ($h\ket{\Psi_s}$) and $p_k=||A_k g\ket{\Psi_s}||^2/||g\ket{\Psi_s}||^2$ (hence $p_k> 0$ and $\sum_k p_k=1$).
Then the outcome is deterministic iff
\begin{equation}\label{Eq_output}
\frac{1}{\sqrt{p_1}}A_1 g_1 \bigotimes_i g_i\ket{\Psi_s}= \frac{1}{\sqrt{p_k}}U_k^{(1)} A_k g_1  \bigotimes_i U_k^{(i)} g_i\ket{\Psi_s},
\end{equation}
$\forall k>1$ and where $U_k=U_k^{(1)}\otimes \ldots \otimes U_k^{(n)}$ is unitary. Using the fact that the only local operations which leave $\ket{\Psi_s}$ invariant are the elements of $S_\Psi$ and the fact that all the local matrices occurring in the expression above must be invertible \footnote{Note that otherwise one would consider different SLOCC classes.}, these conditions are equivalent to the existence of local unitaries $U_k$ and symmetries $S_k\in S_\Psi$ such that $\forall k$
\bi \item[(a)] $A_1g_1= r_k^{(1)} U_k^{(1)} A_k g_1 S_k^{(1)} $
\item[(b)] $g_i=r_k^{(i)} U_k^{(i)} g_i S_k^{(i)}  $ $\forall i>1$.\ei
Here, $\prod_{i=1}^n  r_k^{(i)}=\sqrt{p_1/p_k}$ for any $k$.
Taking the determinant on the right and left hand side of the equation in (b) and taking into account that both $U_i^{(k)}$ and $S_i^{(k)}$ are unitary leads to $|r_k^{(i)}|=1$ for any $k$ and $i>1$. Hence, $G_i=(S_k^{(i)})^\dagger G_i S_k^{(i)}$ or equivalently $[G_i,S_k^{(i)}]=0$ $\forall i>1$ and $\forall k$. Using now the equation in (a) we obtain
\begin{equation}
g_1^\dagger A_1^\dagger A_1g_1= \frac{p_1}{p_k} (S_k^{(1)})^\dagger g_1^\dagger A_k^\dagger A_k g_1 S_k^{(1)},
\end{equation}
where we have used that $|r_k^{(1)}|^2 =p_1/p_k$ (given that we found above that $|r_k^{(i)}|=1$ for any $k$ and $i>1$). Defining the positive and normalized operator
\begin{equation}\label{Eqconv2}
H=\frac{g_1^\dagger A_1^\dagger A_1g_1}{\tr(g_1^\dagger A_1^\dagger A_1g_1)}
\end{equation}
and using that $\sum_k A_k^\dagger A_k =\one$, we finally obtain $G_1=\sum_k p_k S_k H S_k^\dagger$ (where $S_1=\one$). Using that the Hermitian operators $G_i$ commute with $S^{(i)}$ iff they commute with $(S^{(i)})^\dagger$ and redefining $S_k$ by $S_k^\dagger$ leads to the conditions $(i)$ and $(ii)$ in the statement of the lemma.

Let us now show the sufficiency of the conditions. In order to see that states $g \ket{ \Psi_s}$ fulfilling conditions (i) and (ii) are convertible via $ LOCC_1$ consider the following protocol. Party 1 implements the measurement $ \{A_k=  \sqrt{p_k} hS_kg_1^{-1} \}_{k=1}^m$
which is a valid POVM due to condition (ii). Then the other parties apply depending on the measurement outcome the LUs $U_k^{(i)}$ which fulfill $U_k^{(i)} g_i=  g_i S_k^{(i)}$. Note that due to condition (i) these LUs have to exist. Applying this  $LOCC_1$
protocol to a state $g  \ket{ \Psi_s}$ leads to a final state $h\otimes g_2 \otimes \ldots \otimes g_n\ket{\Psi_s}$ where $H=h^ \dagger h$ fulfills condition (ii) (involving only the symmetries for which condition (i) holds).
In fact it can be easily seen \footnote{Using Eq. (\ref{Eqconv2}) one obtains that $A_1  \propto W hg_1^{-1},$ where $h^ \dagger h=H$ and W is a unitary. Moreover, as shown before for a transformation to be possible $H$ has to fulfill that $G_1=\sum_{k=1}^m p_k (S_k^{(1)})^\dagger H S_k^{(1)}$ with the symmetries fulfilling condition (i). With this it follows directly from Eq. (\ref{Eqconv1}) that the final states of this transformation have to be of the form $h\otimes g_2 \otimes \ldots \otimes g_n\ket{\Psi_s}$ where $H=h^ \dagger h$ fulfills condition (ii) for symmetries for which condition (i) holds.} that these are the only possible final states of a $LOCC_1$ protocol.
Note that the states $g  \ket{ \Psi_s}$ and $h\otimes g_2 \otimes \ldots \otimes g_n\ket{\Psi_s}$ are LU-eqivalent iff there exist a symmetry $S\in S_{\Psi_s}$ fulfilling condition (i) and $H=S^{(1)} G_1 (S^{(1)})^\dagger$.
\end{proof}

\section{Proof of Lemma \ref{lem:2mstates}\label{lemma6}}

This appendix is devoted to the proof of Lemma \ref{lem:2mstates}, which will be done inductively.
That is, we prove that for generic normalized $\vec{\alpha} \in \C^4$ the stabilizer of the recursively defined $2^m$-qubit state
 \begin{align*}
 \ket{\psi_m(\vec{\alpha})} \propto \frac{1}{2}(\one + U_{0,2^{m-1}}) \ket{\psi_{m-1}(\vec{\alpha})}^{\otimes 2}, \ m > 2,
\end{align*}
is $S_{\psi_m(\vec{\alpha})} = \{\sigma_i^{\otimes 2^m}\}_{i=0}^3$ for almost every $\vec\alpha$, where the qubits are labeled from $0$ to $2^m$ and $U_{i,j}$ permutes particles $i$ and $j$.
We continue with the notation introduced in the paragraph after Lemma \ref{lem:2mstates} in the main text. There,
we commented in detail on the way the states $\ket{\psi_m(\vec{\alpha})}$ can be construced from many copies of states $\ket{\psi_k(\vec{\alpha})}$, with $k<m$. Recall that
\begin{align}
  \ket{\psi_m(\vec{\alpha})} \propto K(m,k) \ket{\psi_k(\vec{\alpha})}^{\otimes 2^{m-k}}, \label{eq:KtoM2}
\end{align}
as we have seen in the main text (see Eq. (\ref{eq:KtoM})).
First, we state and prove two observations that we use in the proof of Lemma \ref{lem:2mstates}.

\begin{observation}
\label{obs:1}
 If $S_{\psi_{k}(\vec{\alpha})} = \{\sigma_i^{\otimes 2^k}\}_{i=0}^3$ for any $k \in \{2,\ldots,m-1\}$, any local symmetry of $\ket{\psi_{m}(\vec{\alpha})}$ has to be of the form
 $S = \bigotimes_{l=0}^{2^{m-2}} (S^{(4l)} \otimes S_1 \otimes S_2 \otimes S_3)$.
\end{observation}

\proof{
In the following we consider $k \in \{2,\ldots,m-1\}$ and denote by $U_k$ the operator that permutes particles $i$ and $2^k+i$ for any $i \in \{1,\ldots,2^k-1\}$.
We first show that
\begin{align}
 U_k\ket{\psi_m(\vec{\alpha})} = \ket{\psi_m(\vec{\alpha})}. \label{eq:Usymm}
\end{align}
That this is the case, can be seen by looking at the structure of $\ket{\psi_m(\vec{\alpha})}$ as illustrated in Fig. \ref{fig:1} and using that
the projector on the symmetric subspace, $K(k+1, k) = 1/2(\one + U_{0,2^k})$ is invariant under the exchange of the particles with index $0$ and $2^k$.
More precisely, in order to prove Eq. (\ref{eq:Usymm}) note first that $[U_k,K(m,k)]=0$ as these operators act on different subspaces. Moreover, $U_k$ applied to
$2^{m-k}$ copies of $\ket{\psi_k(\vec{\alpha})}$ acts
nontrivially only on two copies of $\ket{\psi_k(\vec{\alpha})}$. As the operator
$U_{0,2^k}U_k$ simply permutes these two copies of $\ket{\psi_k(\vec{\alpha})}$ it holds that
\begin{align}
 U_{0,2^k}U_k \ket{\psi_k(\vec{\alpha})}^{\otimes 2^{m-k}} =  \ket{\psi_k(\vec{\alpha})}^{\otimes 2^{m-k}},
\end{align}
which is equivalent to
\begin{align}
 U_k \ket{\psi_k(\vec{\alpha})}^{\otimes 2^{m-k}} = U_{0,2^k}\ket{\psi_k(\vec{\alpha})}^{\otimes 2^{m-k}}.
\end{align}
Finally, note
that $K(m,k) = \tilde{K}(m,k) 1/2(\one + U_{0,2^k})$, for an operator $\tilde{K}(m,k)$. Hence, we have $K(m,k)U_{0,2^k} = K(m,k)$. Using all these relations it is straightforward to
see that,
\begin{align}
 &U_k\ket{\psi_m(\vec{\alpha})} \propto U_k K(m,k) \ket{\psi_k(\vec{\alpha})}^{\otimes 2^{m-k}} \nonumber\\
 &= K(m,k) U_k \ket{\psi_k(\vec{\alpha})}^{\otimes 2^{m-k}} \nonumber\\
 &= K(m,k)\ket{\psi_k(\vec{\alpha})}^{\otimes 2^{m-k}} \propto \ket{\psi_m(\vec{\alpha})}, \label{eq:Ul}
\end{align}
which proves Eq. (\ref{eq:Usymm}).
For a local symmetry
$S = S^{(1)} \otimes \ldots \otimes S^{(2^m)} \in S_{\psi_m(\vec{\alpha})}$ we therefore have that $U_kSU_k^\dagger \in S_{\psi_m(\vec{\alpha})}$. Moreover, as $S^{-1} \in S_{\psi_m(\vec{\alpha})}$
we have
\begin{align}
U_kSU_k^\dagger S^{-1} = \one_2 \otimes Y_k \otimes \one_2 \otimes Y_k^{-1} \otimes \one_{2^m-2^{k+1}} \in S_{\psi_m(\vec{\alpha})}, \label{eq:symmY}
\end{align}
where $Y_k = S^{(2^k+1)}(S^{(1)})^{-1} \otimes \ldots \otimes S^{(2^{k+1}-1)}(S^{(2^{k}-1)})^{-1}$ and where the index of the identity operators indicate the dimension of the subspace on which they act on.\\

In the following we express for any $l \in \{0,...,2^{m-k}-1\}$ the $2^k$-qubit state $\ket{\psi_k(\vec{\alpha})}_{2^kl,\ldots,2^k(l+1)-1}$ of qubits with indices $2^kl$ to $2^k(l+1)-1$
in the Schmidt decomposition of its first particle, i.e. the particle with index $2^kl$, with
the rest, i.e.
\begin{align*}
 \ket{\psi_k(\vec{\alpha})}_{2^kl,\ldots,2^k(l+1)-1} = d_1 \ket{a_1}_{2^k l}\ket{\phi_1}_{\vec{l}}+d_2 \ket{a_2}_{2^k l}\ket{\phi_2}_{\vec{l}}.
\end{align*}
Note that the states occuring on the right-hand-side of the equation above depend on the value of $k$. However, as we never consider in the following this decomposition for different values of $k$, we do not
indicate this dependency explicitely.
For example, we express the state $\ket{\psi_3(\vec{\alpha})}_{0,\ldots,7}$ of the first eight qubits as
\begin{align*}
\ket{\psi_3(\vec{\alpha})}_{0,\ldots,7} = d_1 \ket{a_1}_{0}\ket{\phi_1}_{1,\ldots,7}+d_2 \ket{a_2}_{0}\ket{\phi_2}_{1,\ldots,7}.
\end{align*}
Here the vector $\vec{l}$ in $\ket{\phi_i}_{\vec{l}}$ indicates that this is a state of particles $2^kl+1,\ldots,2^k(l+1)-1$.
Clearly, $d_i \neq 0$ as the $2^k$-qubit state is multipartite entangled and $\bk{a_i|a_j} = \bk{\phi_i|\phi_j} = \delta_{i,j}$.\\

Let us now define the $(2^m-2^{k+1}+2)$-qubit state
\begin{align*}
 \ket{\Phi_{i,j}} = \ket{a_i}_0 \ket{a_j}_{2^k} \bigotimes_{l=2}^{2^{m-k}-1} \ket{a_i}_{2^k l} \ket{\phi_i}_{\vec{l}}, \ \text{for} \ i,j \in \{0,1\}.
\end{align*}
Note that this state does not include particles $1,\ldots,2^k-1$ and $2^k+1,\ldots,2^{k+1}-1$. Next, we project both sides of the eigenvalue equation $U_kSU_k^\dagger S^{-1}\ket{\psi_m(\vec{\alpha})} = \ket{\psi_m(\vec{\alpha})}$ on the state $\ket{\Phi_{i,j}}$, i.e. we get

\begin{align}
 \bk{\Phi_{i,j}|U_kSU_k^\dagger S^{-1}|\psi_m(\vec{\alpha})} = \bk{\Phi_{i,j}|\psi_m(\vec{\alpha})}. \label{eq:EVEq0}
\end{align}

In order to simplify this expression recall that
$K(m,k)$ in Eq. (\ref{eq:KtoM2}) can be expressed as $K(m,k) = \tilde{K}(m,k) \frac{1}{2}(\one+U_{0,2^k})$ for an operator $\tilde{K}(m,k)$ that does not act on the particle with index $2^k$ and which fulfills
$\tilde{K}(m,k) \ket{\Phi_{i,j}} \propto \ket{\Phi_{i,j}}$.
\begin{widetext}

Using the properties of $K(m,k)$ it is easy to see that,
\begin{align}
 K(m,k)\ket{\Phi_{i,j}} \propto \frac{1}{2}(\one+U_{0,2^k})\ket{\Phi_{i,j}} \propto (\ket{a_i}_0 \ket{a_j}_{2^k} + \ket{a_j}_0 \ket{a_i}_{2^k}) \bigotimes_{l=1}^{2^{m-k}-1}  \ket{a_i}_{2^k l} \ket{\phi_i}_{\vec{l}}.
\end{align}
Using this relation, a straightforward calculation reveals that the right-hand-side of Eq. (\ref{eq:EVEq0}) simplifies as follows,
\begin{align*}
 &\bk{\Phi_{i,j}|\psi_m(\vec{\alpha})} =  \bk{\Phi_{i,j}|K(m,k)|\psi_k(\vec{\alpha})}^{\otimes 2^{m-k}}
  \propto (\bra{a_i}_0 \bra{a_j}_{2^k} + \bra{a_j}_0 \bra{a_i}_{2^k}) \bigotimes_{l=1}^{2^{m-k}-1} \bra{a_i}_{2^k l} \bra{\phi_i}_{\vec{l}} \ket{\psi_k(\vec{\alpha})}^{\otimes 2^{m-k}}\\
  &\propto \ket{\phi_i}_{\vec{0}} \ket{\phi_j}_{\vec{1}} + \ket{\phi_j}_{\vec{0}} \ket{\phi_i}_{\vec{1}}.
\end{align*}
\end{widetext}

Note that the symmetry $U_kSU_k^\dagger S^{-1}$ acts trivially on all particles on which $\ket{\Phi_{i,j}}$ is defined (see Eq. (\ref{eq:symmY})).
Hence, the left-hand-side of Eq. (\ref{eq:EVEq0}) can be easily evaluated to obtain
\begin{align*}
 &(Y_k \otimes Y_k^{-1})(\ket{\phi_i}_{\vec{0}} \ket{\phi_j}_{\vec{1}} + \ket{\phi_j}_{\vec{0}} \ket{\phi_i}_{\vec{1}}) \\
 &\propto (\ket{\phi_i}_{\vec{0}} \ket{\phi_j}_{\vec{1}} + \ket{\phi_j}_{\vec{0}} \ket{\phi_i}_{\vec{1}}), \ i,j \in \{0,1\}.
\end{align*}
We hence obtain $Y_k\ket{\phi_i}=\lambda_i \ket{\phi_i}$ for $i = 0,1$ and some $\lambda_i \neq 0$. For $i \neq j$ we
subsequently get $\lambda_0 = \lambda_1 = \lambda \neq 0$. Using these relations we obtain the following eigenvalue equation,
\begin{align}
 &(\one \otimes Y_k)\ket{\psi_k(\vec{\alpha})} = d_1 \ket{a_1} Y_k\ket{\phi_1}+d_2 \ket{a_2} Y_k\ket{\phi_2} \nonumber\\
 &= d_1 \ket{a_1}\lambda\ket{\phi_1}+d_2 \ket{a_2}\lambda\ket{\phi_2} = \lambda \ket{\psi_k(\vec{\alpha})}. \label{eq:EVEq1}
\end{align}
Recall that the requirement of the Observation is that $\ket{\psi_k(\vec{\alpha})}$, for $2 \leq k < m$, has only symmetries of the form $\{\sigma_i^{\otimes 2^k}\}_{i=0}^3$ and hence Eq. (\ref{eq:EVEq1}) can only be
satisfied if $Y_k \propto \one$. Looking at the definition of
$Y_k$ we see that this yields $S^{(j)} \propto S^{(2^k+j)}$, for $j \in \{1,\ldots,2^k-1\}$ and we can set w.l.o.g.
\begin{align}
 S^{(j)} = S^{(2^k+j)}, \ \forall j \in \{1,\ldots,2^k-1\}, 2 \leq k < m. \label{eq:SymEq}
\end{align}
For $k=2$ we see from Eq. (\ref{eq:SymEq}) that $S^{(1)} = S^{(5)} = S_1, S^{(2)} = S^{(6)} = S_2$ and $S^{(3)} = S^{(7)} = S_3$ for some invertible operators $S_1, S_2$ and $S_3$.
For $k=3$ it is then easy to see that
also $S^{(9)} = S^{(13)} = S_1, S^{(10)} = S^{(14)} = S_2$ and $S^{(11)} = S^{(15)} = S_3$, etc. Hence, we obtain
\begin{align*}
 S^{(4l+1)} = S_1, S^{(4l+2)} = S_2 \ \text{and}\ S^{(4l+3)} = S_3
\end{align*}
for all $ l \in \{0,...,2^{m-2}-1\}$, which proves the observation. \qed
}\\

Before we state the second observation needed in the proof of Lemma \ref{lem:2mstates} we define for $\ket{\psi} \in (\C^2)^{\otimes 4}$ the maps
\begin{align*}
 \Lambda_1(\psi) = \bra{\psi^-}_{2,6}\bra{\psi^-}_{3,7} \frac{1}{2}(\one + U_{0,4})\ket{\psi}^{\otimes 2},\\
 \Lambda_2(\psi) = \bra{\psi^-}_{1,5}\bra{\psi^-}_{3,7} \frac{1}{2}(\one + U_{0,4})\ket{\psi}^{\otimes 2},\\
 \Lambda_3(\psi) = \bra{\psi^-}_{1,5}\bra{\psi^-}_{2,6} \frac{1}{2}(\one + U_{0,4})\ket{\psi}^{\otimes 2},
\end{align*}
where $\ket{\psi^-} = 1/\sqrt{2} (\ket{01}-\ket{10})$ is the singlet state.
It is easy to show that the maps
\begin{align}
 &\lambda_i: \C^4 \rightarrow \C^4 \label{eq:lambda}\\
 &\vec{\alpha} \mapsto \lambda_i(\vec{\alpha}), \ s.t. \ \frac{\Lambda_i(\psi_2(\vec{\alpha}))}{\|\Lambda_i(\psi_2(\vec{\alpha}))\|} = \ket{\psi_2(\lambda_i(\vec{\alpha}))} \nonumber
\end{align}
are well-defined for $i=1,2,3$. Using the notation $(\lambda_i \circ \lambda_j)(\vec{\alpha}) = \lambda_i(\lambda_j(\vec{\alpha}))$ we can state the next observation, which implies that certain
projections of $\ket{\psi_{m+1}(\vec{\alpha})}$ onto singlet states lead to a state
$\ket{\psi_m(\vec{\gamma})}$ which is constructed from a generic four-qubit state $\ket{\psi_2(\vec{\gamma})}$, even if this procedure is applied recursively.

\begin{observation}
\label{obs:2}
For almost all normalized $\vec{\alpha} \in \C^4$, any integer $N$ and any $\vec{k} = (k_1,\ldots,k_N) \in \{1,2,3\}^N$ it holds that $S_{\psi_2((\bigcirc_{i=1}^N\lambda_{k_i})(\vec{\alpha}))} = \{\sigma_i^{\otimes 4}\}_{i=0}^3$.
\end{observation}

\proof{
 It is easy to see that, for almost every $\vec{\alpha} \in \C^4$, $\ket{\psi_2(\lambda_i(\vec{\alpha}))}$ is a generic element of the following set of four-qubit states
\begin{widetext}
\begin{align*}
 \mathcal{P} \equiv \left \{\sum_{i=0}^3 \beta_i \sigma_i^{(0)} \sigma_i^{(1)} \ket{\phi^+}_{02}\ket{\phi^+}_{13} \ s.t. \ \beta_i \in \C, \sum_{i=0}^3 |\beta_i|^2 = 1, \beta_2 = 0,\beta_i^2 \neq \beta_j^2 \ \text{for} \ i \neq j, \nexists \ q \in \C\setminus\{1\} \ \text{s.t.} \ \{\beta_i^2\}_{i=0}^3 = \{q \beta_i^2\}_{i=0}^3 \right \},
\end{align*}
\end{widetext}
for $i = 1,2,3$. It has been shown in \cite{SpdV16} that the stabilizer of the states in $\mathcal{P}$ is $\{\sigma_i^{\otimes 4}\}_{i=0}^3$.
Thus, $S_{\psi_2(\lambda_{k_1}(\vec{\alpha}))} = \{\sigma_i^{\otimes 4}\}_{i=0}^3$, which
proves the observation for $N = 1$. It is easy to show that $\mathcal{P} \subset \Lambda_i(\mathcal{P})$ and that $\Lambda_i$ maps almost all elements of $P$ (up to normalization) to an element of $P$, for $i=1,2,3$.
Hence, the state $\ket{\psi_2((\bigcirc_{i=1}^N\lambda_{k_i})(\vec{\alpha}))} \propto (\bigcirc_{i=2}^N\Lambda_{k_i})(\psi_2(\lambda_{k_1}(\vec{\alpha})))$ is, for almost every normalized $\vec{\alpha} \in \C^4$, an element of $\mathcal{P}$ with a stabilizer of the form $\{\sigma_i^{\otimes 4}\}$. This completes the proof of the observation. \qed
}\\

We can now use these observations to prove Lemma \ref{lem:2mstates}, which we state here again.\\

\noindent {\bf Lemma 6.} \hspace{-0.35cm}
\emph{
 For almost every normalized $\vec{\alpha} \in \C^4$ the stabilizer of the recursively defined $2^m$-qubit state
 \begin{align*}
 \ket{\psi_m(\vec{\alpha})} \propto \frac{1}{2}(\one + U_{0,2^{m-1}}) \ket{\psi_{m-1}(\vec{\alpha})}^{\otimes 2}, \ m > 2,
\end{align*}
 is $S_{\psi_m(\vec{\alpha})} = \{\sigma_i^{\otimes 2^m}\}_{i=0}^3$.
}

\proof{We prove the lemma by induction over $m$. That is, we first show that it is correct for $m=3$, i.e. $S_{\psi_3(\vec{\alpha})} = \{\sigma_i^{\otimes 8}\}_{i=0}^3$, and then we show that
$S_{\psi_{m+1}(\vec{\alpha})} = \{\sigma_i^{\otimes 2^{m+1}}\}_{i=0}^3$ if $S_{\psi_k(\vec{\alpha})} = \{\sigma_i^{\otimes 2^k}\}_{i=0}^3$ for almost all normalized $\vec{\alpha} \in \C^4$ and for all $2 \leq k < m+1$.\\

Let us first show that the lemma is correct for $m=3$. Using that $S_{\psi_2(\vec{\alpha})} = \{\sigma_i^{\otimes 4}\}_{i=0}^3$ for generic $\vec{\alpha} \in \C^4$,
Observation \ref{obs:2} implies that any $S \in S_{\psi_3(\vec{\alpha})}$ is of the form $S = \bigotimes_{l=0}^{1} (S^{(4l)} \otimes S_1 \otimes S_2 \otimes S_3)$.
Using that
$X \otimes X\ket{\psi^-} \propto \ket{\psi^-}$ for any invertible $X$ the eigenvalue equation
$S\ket{\psi_3(\vec{\alpha})} = \ket{\psi_3(\vec{\alpha})}$ projected on $\ket{\psi^-}_{1,5}\ket{\psi^-}_{2,6}$ leads to
\begin{align}
 &(S^{(0)} \otimes S_3 \otimes S^{(4)} \otimes S_3)\bra{\psi^-}_{1,5}\bra{\psi^-}_{2,6}\ket{\psi_3(\vec{\alpha})} \nonumber \\
 &\propto \bra{\psi^-}_{1,5}\bra{\psi^-}_{2,6}\ket{\psi_3(\vec{\alpha})} \label{eq:EVEq3}.
\end{align}
It is easy to see that
\begin{align*}
\bra{\psi^-}_{1,5}\bra{\psi^-}_{2,6}\ket{\psi_3(\vec{\alpha})} \propto \Lambda_3(\psi_2(\vec{\alpha})) \propto \ket{\psi_2(\lambda_3(\vec{\alpha}))},
\end{align*}
and hence Observation \ref{obs:2}
states that $S_{\psi_2(\lambda_3(\vec{\alpha}))} = \{\sigma_i^{\otimes 4}\}$ for generic $\vec{\alpha}$.\\

Combined with Eq. (\ref{eq:EVEq3}) we obtain $S^{(0)} \otimes S_3 \otimes S^{(4)} \otimes S_3 \propto \sigma_i^{\otimes 4}$. Analogously, we can apply $\Lambda_2, \Lambda_1$ to obtain that any symmetry
has to be proportional to some $\sigma_i^{\otimes 4}$ on the remaining parties. Combining all these facts we have that $S_{\psi_3(\vec{\alpha})}  = \{\sigma_i^{\otimes 8}\}_{i=0}^3$ and
therefore the lemma is proven for $m=3$.\\

Let us complete the prove by showing that the lemma is correct for $m+1$, if it is valid for
all $k < m+1$. Due to Observation \ref{obs:1} we know that any symmetry of $\ket{\psi_{m+1}(\vec{\alpha})}$ is of the form
$S = \bigotimes_{l=0}^{2^{m-1}} (S^{(4l)} \otimes S_1 \otimes S_2 \otimes S_3)$.
In order to show now that the stabilizer of $\ket{\psi_{m+1}(\vec{\alpha})}$ is $\{\sigma_i^{\otimes 2^{m+1}}\}_{i=0}^3$, we
proceed analogously to before. We project half of the qubits on both sides of the eigenvalue equation $S\ket{\psi_{m+1}(\vec{\alpha})} = \ket{\psi_{m+1}(\vec{\alpha})}$ onto singlet states
to relate the case $m+1$ to the case $m$. More precisely, we project onto the state
\begin{align}
 \ket{\Phi} = \bigotimes_{l=0}^{2^{m-2}-1} \ket{\psi^-}_{8l+1,8l+5}\ket{\psi^-}_{8l+2,8l+6},
\end{align}
and obtain the eigenvalue equation
\begin{align}
\bk{\Phi|S|\psi_{m+1}(\vec{\alpha})} = \bk{\Phi|\psi_{m+1}(\vec{\alpha})}. \label{eq:EVEq2}
\end{align}

\begin{widetext}
As shown in the proof of Observation \ref{obs:2}, we have
\begin{align*}
 &\bra{\psi^-}_{8l+1,8l+5}\bra{\psi^-}_{8l+2,8l+6}\ket{\psi_{3}(\vec{\alpha})}_{8l,\ldots,8l+7} \propto \Lambda_3(\psi_2(\vec{\alpha}))_{8l,\ldots,8l+7} \propto \ket{\psi_{2}(\lambda_3(\vec{\alpha}))}_{8l,8l+3,8l+4,8l+7}.
\end{align*}
Hence, Eq. (\ref{eq:EVEq2}) can be written as
\begin{align}
 &\bk{\Phi|S|\psi_{m+1}(\vec{\alpha})} = \bk{\Phi|SK(m+1,3)|\psi_{3}(\vec{\alpha})}^{\otimes 2^{m-2}} \propto \tilde{S} \bk{\Phi|K(m+1,3)|\psi_{3}(\vec{\alpha})}^{\otimes 2^{m-2}} \\
 & \propto \tilde{S} K(m+1,3) \ket{\psi_{2}(\lambda_3(\vec{\alpha}))}^{\otimes 2^{m-2}} \propto \tilde{S}\ket{\psi_m(\lambda_3(\vec{\alpha}))} = \ket{\psi_m(\lambda_3(\vec{\alpha}))}, \label{eq:RedEVEq}
\end{align}
where $\tilde{S}$ is a local symmetry on the remaining particles.
\end{widetext}
Here, we used the easily verifiable identity (see in Fig. \ref{fig:1} and Fig. \ref{fig:2})
\begin{align}
 K(m+1,3) \ket{\psi_{2}(\lambda_3(\vec{\alpha}))}^{\otimes 2^{m-2}} = \ket{\psi_m(\lambda_3(\vec{\alpha}))}.
\end{align}

From Eq. (\ref{eq:RedEVEq}) we have that the symmetry on the remaining particles, $\tilde{S}$, has to fulfill
\begin{align}
 \tilde{S}\ket{\psi_m(\lambda_3(\vec{\alpha}))} = \ket{\psi_m(\lambda_3(\vec{\alpha}))}, \label{eq:RedEVEq2}
\end{align}
 where, according to Observation \ref{obs:2}, $\lambda_3(\vec{\alpha})$ is such that $\ket{\psi_{2}(\lambda_3(\vec{\alpha}))}$ has the stabilizer $\{\sigma_i^{\otimes 4}\}_{i=0}^3$.
 In an analogous manner we can project on particles $8l+1,8l+5$ and $8l+3,8l+7$ or $8l+2,8l+6$ and $8l+3,8l+7$, for $l \in \{0,...,2^{m-2}-1\}$ and obtain the states
 $\ket{\psi_m(\lambda_2(\vec{\alpha}))}$ and $\ket{\psi_m(\lambda_1(\vec{\alpha}))}$, respectively. Before we use the induction assumption we have to make sure that we perform a
 valid inductive step if we reduce the state $\ket{\psi_{m+1}(\vec{\alpha})}$ to the states $\ket{\psi_m(\lambda_i(\vec{\alpha}))}$, $i = 1,2,3$. This is the case only if these reductions can also be applied to
 the states $\ket{\psi_m(\lambda_i(\vec{\alpha}))}$ themselves. It is easy to see that this condition is equivalent to the fact that the initial parameters $\vec{\alpha}$ are such that
 \begin{align}
 S_{\psi_2((\bigcirc_{i=1}^N\lambda_{k_i})(\vec{\alpha}))} = \{\sigma_i^{\otimes 4}\}_{i=0}^3
 \end{align}
 for all $1 \leq N \leq m$ and all $\vec{k} = (k_1,\ldots,k_N) \in \{1,2,3\}^N$. For almost every normalized $\vec{\alpha}$ this is, however, precisely the content of
 Observation \ref{obs:2}. We can therefore use the induction assumption that $\ket{\psi_m(\lambda_3(\vec{\alpha})}$
 only has symmetries of the form $\sigma_i^{\otimes 2^m}$ and hence Eq. (\ref{eq:RedEVEq2}) is only fulfilled if $\tilde{S} \propto \sigma_{i_3}^{\otimes 2^m}$ for some $i_3 \in \{0,1,2,3\}$. We can draw
 analogous conclusions for the projections on the other particles, i.e. we obtain that, on the remaining particles that remain after these projections, the symmetry is of the form $\sigma_{i_k}^{\otimes 2^m}$ for some $i_k \in \{0,1,2,3\}$ and $k = 1,2$.
 Note that the particle with index $0$ is a remaining particle for all of these projections and hence $i_1 = i_2 = i_3$. We therefore obtain
 $S \in \{\sigma_i^{\otimes 2^{m+1}}\}_{i=0}^3$, i.e.
 $S_{\psi_{m+1}(\vec{\alpha})} = \{\sigma_i^{\otimes 2^{m+1}}\}_{i=0}^3$, if
$S_{\psi_{k}(\vec{\alpha})} = \{\sigma_i^{\otimes 2^{k}}\}_{i=0}^3$ for all $2\leq k < m+1$ and almost every normalized $\vec{\alpha} \in \C^4$. This completes the proof. \qed

\begin{figure}
 \includegraphics[width=0.45\textwidth]{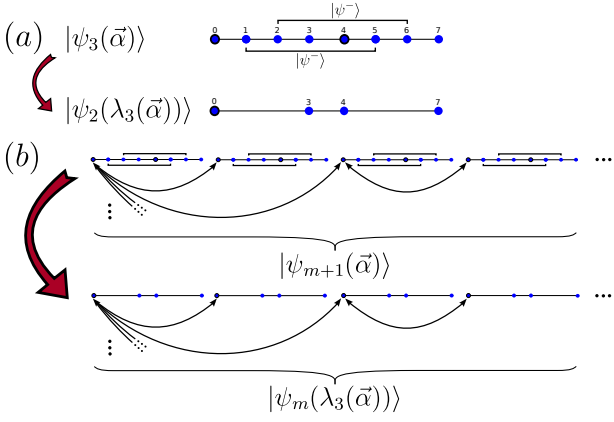}
 \caption{This figure depicts the projection operation that are used in the proof of Lemma \ref{lem:2mstates}. In (a) one can see how the eight-qubit state
 $\ket{\psi_3(\vec{\alpha})}$ is mapped on a four-qubit state $\ket{\psi_2(\lambda_3(\vec{\alpha}))}$ by projecting each of the pairs of qubits $1,5$ and $2,6$ on the singlet state
 $\ket{\psi^-}$. In (b) the result of (a) is used to map $\ket{\psi_{m+1}(\vec{\alpha})}$ to $\ket{\psi_m(\lambda_3(\vec{\alpha}))}$. We build the former state from $2^{m+1-3}$ copies
 of $\ket{\psi_3(\vec{\alpha})}$ on which the operation $K(m+1,3)$ acts. Then projections on the qubits $1,5$ and $2,6$ of each eight-qubit state are
 performed as shown in (b). Using the result of (a) it is easy to see
 that the resulting global state is $\ket{\psi_m(\lambda_3(\vec{\alpha}))}$.}
 \label{fig:2}
\end{figure}
}

\end{appendix}


\begin{thebibliography}{widest-label}

\bibitem{nielsenchuang} M.A. Nielsen, and I.L. Chuang, \textit{Quantum Computation and Quantum Information} (Cambridge University Press, New York, 2000).

\bibitem{reviews} See e.\ g.\ R. Horodecki \textit{et al.}, Rev. Mod. Phys. \textbf{81}, 865 (2009) and references therein.

\bibitem{RaBr01} R. Raussendorf and H.J. Briegel, Phys. Rev. Lett. \textbf{86}, 5188 (2001).

\bibitem{reviewmet} V. Giovannetti, S. Lloyd, and L. Maccone, Science \textbf{306}, 1330-1336 (2004).

\bibitem{SecretSh} M. Hillery, V. Bu\v{z}ek, and A. Berthiaume Phys. Rev. A \textbf{59}, 1829 (1999); D. Gottesman Phys. Rev. A 61, 042311 (2000) and references therein.

\bibitem{AmFa08} For a review see e.g. L. Amico, R. Fazio, A. Osterloh, V. Vedral, Rev. Mod. Phys. {\bf 80}, 517
(2008) and references therein.

\bibitem{orus} See e.g. R. Orus, Ann. Phys. \textbf{349}, 117 (2014) and references therein.

\bibitem{chitambar1} See e.\ g.\ M. J. Donald, M. Horodecki, and O. Rudolph, J. Math. Phys., {\bf 43}, 4252 (2002); E. Chitambar, W. Cui, and H.-K-. Lo, Phys. Rev. Lett. \textbf{108}, 240504 (2012); E. Chitambar, D. Leung, L. Mancinska, M. Ozols, and A. Winter, Commun. Math. Phys. \textbf{328}, 303 (2014).

\bibitem{chit11} E. Chitambar, Phys. Rev. Lett. \textbf{107}, 190502 (2011).

\bibitem{rains} E. M. Rains, arXiv:quant-ph/9707002 (1997).

\bibitem{sepnotlocc} C. H. Bennett, D. P. DiVincenzo, C. A. Fuchs, T. Mor, E. Rains, P. W. Shor, J. A. Smolin, and W. K. Wootters, Phys. Rev. A \textbf{59}, 1070 (1999); M. Kleinmann, H. Kampermann, and D. Bru{\ss}, Phys. Rev. A \textbf{84}, 042326 (2011).

\bibitem{LoPopescu} H.-K. Lo and S. Popescu, Phys. Rev. A \textbf{63}, 022301 (2001).

\bibitem{nielsen} M. A. Nielsen, Phys. Rev. Lett. \textbf{83}, 436 (1999).

\bibitem{gheorghiu} V. Gheorghiu and R. B. Griffiths, Phys. Rev. A \textbf{78}, 020304 (R) (2008).

\bibitem{turgutghz} S. Turgut, Y. G\"ul, and N.K. Pak, Phys. Rev. A \textbf{81}, 012317 (2010).

\bibitem{turgutw} S. Kintas and S. Turgut, J. Math. Phys. \textbf{51}, 092202 (2010).

\bibitem{dVSp13} J. I. de Vicente, C. Spee, and B. Kraus, Phys. Rev. Lett. \textbf{111}, 110502 (2013).

\bibitem{SaSc15} D. Sauerwein, K. Schwaiger, M. Cuquet, J. I. de Vicente, and B. Kraus, Phys. Rev. A \textbf{92}, 062340 (2015).

\bibitem{SpdV16} C. Spee, J. I. de Vicente, and B. Kraus, J. Math. Phys. \textbf{57}, 052201 (2016).

\bibitem{HeSp15} M. Hebenstreit, C. Spee, and B. Kraus, Phys. Rev. A \textbf{93}, 012339 (2016).

\bibitem{Gour} G. Gour and N.R. Wallach, New J. Phys. \textbf{13}, 073013 (2011).

\bibitem{cohen} S. M. Cohen, arXiv:1606.00029 (2016).

\bibitem{short} C. Spee, J.I. de Vicente, D. Sauerwein, and B. Kraus, arXiv:1606.04418v1 (2016).

\bibitem{bookWallach} N.R. Wallach, \emph{Geometric invariant theory over the real and complex numbers}, Springer, in preparation.

\bibitem{slocc} W. D\"ur, G. Vidal, and J.I. Cirac, Phys. Rev. A \textbf{62}, 062314 (2000); F. Verstraete, J. Dehaene, B. De Moor, and H. Verschelde, Phys. Rev. A, \textbf{65}, 052112 (2002).

\bibitem{BrLu04} E. Briand, J.-G. Luque, J.-Y. Thibon, and F. Verstraete, J. Math. Phys. \textbf{45}, 4855  (2004).

\bibitem{GoWa} G. Gour and N.R. Wallach, J. Math. Phys. \textbf{51}, 112201 (2010).

\bibitem{GourSpek06} G. Gour, and R. W. Spekkens, Phys. Rev. A \textbf{73}, 062331 (2006).

\bibitem{ScSa15} K. Schwaiger, D. Sauerwein, M. Cuquet, J. I. de Vicente, and B. Kraus, Phys. Rev. Lett. \textbf{115}, 150502 (2015).

\bibitem{mps} D. Perez-Garcia, F. Verstraete, M. M. Wolf, and J. I. Cirac, Quantum Inf. Comput. \textbf{7}, 401 (2007).

\bibitem{VeCi04} F. Verstraete and J. I. Cirac, Phys. Rev. A \textbf{70}, 060302(R) (2004); F. Verstraete and J. I. Cirac, arXiv:cond-mat/0407066 (2004).

\bibitem{GoKr16} G. Gour, B. Kraus, and N. Wallach,  quant-ph/1609.01327 (2016).


\end{thebibliography}
\end{document}